\documentclass[9pt]{article}
\usepackage{amssymb,latexsym,amsfonts,amsmath,diagrams, mathtools, amsthm}
\usepackage{graphicx, xcolor}
\usepackage{tikz}
\usepackage{tikz-cd}
\usetikzlibrary{automata, positioning,matrix,fit}
\usepackage{authblk}

\newtheorem{theorem}{Theorem}[section]

\newtheorem{proposition}[theorem]{Proposition}

\newtheorem{definition}[theorem]{Definition}
\newtheorem{example}[theorem]{Example}

\newtheorem{remark}[theorem]{Remark}
\numberwithin{equation}{section}

\newcommand{\R}{{\mathbb{R}}}

\newcommand{\N}{{\mathbb{N}}}

\newcommand{\mc}{\mathcal}
\newcommand{\mbf}{\mathbf}

\newcommand{\tp}{\texttt{p}}

\newcommand{\cB}{\mathcal{B}}
\newcommand{\cS}{\mathcal{S}}
\newcommand{\cQ}{\mathcal{Q}}

\newcommand{\mcF}{\mathcal{F}}
\newcommand{\mcFi}{\mathcal{F}^{-1}}
\newcommand{\mcS}{\mathcal{S}}
\newcommand{\mcAS}{\mathcal{AS}}
\newcommand{\mcSFi}[1][F]{{\mathcal{S},\mathcal{#1}^{-1}}}
\newcommand{\mcASFi}[1][F]{{\mathcal{AS},\mathcal{#1}^{-1}}}

\usepackage{expl3}
\usepackage{environ}
\ExplSyntaxOn
\seq_new:N \g_appendices_seq% define a sequence for holding the appendices
\NewEnviron{Appendix}{\seq_gput_right:No \g_appendices_seq \BODY}
\newcommand\AddAppendices{% regurgitate the appendices
  \appendix% turn all subsequent sections into appendices
  \seq_map_inline:Nn \g_appendices_seq {##1}
}
\ExplSyntaxOff

\usepackage{etoolbox}
\AtEndDocument{\AddAppendices}

\tikzset{commutative diagrams/.cd,
mysymbol/.style={start anchor=center,end anchor=center,draw=none}
}

\title{A Contract Theory for Layered Control Architectures}
\author[$\dagger$]{Manuel Mazo Jr.}
\author[$\star$]{Will Compton} 
\author[$\star$]{Max H. Cohen}
\author[$\star$]{Aaron D. Ames}
\affil[$\dagger$]{Delft Center for Systems and Control, TU Delft (The Netherlands)}
\affil[$\star$]{Department of Mechanical Engineering, Caltech (USA)}

\begin{document}

\maketitle

\begin{abstract}
Autonomous systems typically leverage layered control architectures with a combination of discrete and continuous models operating at different timescales.  As a result, layered systems form a new class of hybrid systems composed of systems operating on a diverse set of continuous and discrete signals.  
This paper formalizes the notion of a layered (hierarchical) control architecture through a theory of relations between its layers.
This theory enables us to formulate contracts within layered control systems---these define interfaces between layers and isolate the design of each layer, guaranteeing that composition of contracts at each layer results in a contract capturing the desired system-wide specification.  Thus, the proposed theory yields the ability to analyze layered control architectures via a compositional approach. 
% We introduce a theory of relations between systems operating on signals of diverse nature. This theory enables us to formalize contracts in layered (hierarchical) control, in which at each layer models (abstraction) of diverse nature are employed. The proposed contracts define interfaces between layers and isolate the design of each layer, guaranteeing that composition of solutions to each layer contract results in a contract capturing the desired system-wide specification.
\end{abstract}

\section{Introduction}

Modern control systems are characterized by their increasing complexity, due not only to the systems being controlled, but also due to the intricacy of the specifications they must satisfy. To address this challenge, engineers often rely on the divide-and-conquer tenet. Complex designs are divided in smaller subsystem designs which, when properly combined, results in the desired functionality for the overall system.
Contract theory~\cite{Benveniste} provides a framework to reason about such compositional designs. 
Although contracts have shown power to reason about the composition of control systems operating at the same abstraction level, what has been called \textit{horizontal} compositions, see e.g.,~\cite{Alkhatib, Saoud, Sharf, Kim, Shali}, 
the tools available to reason about composition of systems of heterogeneous nature is more limited. To address this heterogeneity, typical of any Cyber-Physical System, and particularly of layered control architectures, the notion of \textit{vertical} contracts is proposed in~\cite{Nuzzo}. The main idea behind vertical contracts is to employ \textit{heterogeneous refinements}, essentially employing functions that provide a semantic mapping between different modelling domains (e.g., continuous/discrete time/space domains). 

Layered, also called hierarchical, control architectures are one example of a system resulting from heterogeneous compositions. In a layered control architecture, different controllers are designed at different levels of abstraction, which, when composed (nested) together, results in a controller for the overall objective. Such architectures are encountered in many complex robotic systems \cite{Ames-Multirate}: at the lowest layer, control happens on a fast time scale with designs employing ODE models and aims at rendering the dynamics more tractable at the middle layer; at the middle layer fast dynamics are ignored, usually discrete time models are employed, enabled by the lower layer fast control, and the goal is to plan trajectories connecting various regions of interest; at the top layer the goal is to plan which regions the robot must visit and in what order, where strategies are designed using finite state machine (FSM) models. 
% \todo{May want to provide a few more remarks here motivating \emph{why} we need to employ layered architectures to solve this problem, e.g., system is continuous time/space, high-dimensional, can't directly use standard tools to solve the problem without breaking it down into more tractable sub-problems that leverage different model representations.}

In this manuscript, we formalize vertical contracts for layered control architectures. Our approach follows, at a conceptual level, the same idea in~\cite{Nuzzo}, introducing first mappings between modelling domains, which we term \textit{transducers}, by analogy with physical devices. Next, we introduce a Mealy-type version of generalized labeled transition systems (GLTS), and a concept of heterogeneous (alternating) simulation relations that, by means of the introduced transducers, allow for comparing models at different levels of abstraction, e.g. compare ODEs with difference equations or FSMs. After establishing some basic properties of these relations, we introduce a version of system composition that captures both interconnections of systems akin to feedforward and feedback composition of plants and controllers with external reference signals. Finally, we show employing the introduced heterogeneous relations how to transfer properties of controller designs on a higher level of abstraction to their application on the systems modelled at a lower level of abstraction. Ultimately, the developed tools allow to precisely formalize vertical contracts for layered controllers. We introduce a contract template for each layer, so that designers can devise controllers for each layer, on models at the corresponding level of abstraction without worrying about the specifics of designs at the other layers. We employ assume-guarantee contracts, in which GLTS combined with transducers establish a form of interface between layers. If each control layer satisfies its contract, then the composition of all of these different controllers, essentially nesting them, results in the desired system-wide properties. 
We illustrate the ideas presented on a classical robotic problem specified at the highest level of abstraction through a temporal logic specification.

\section{Preliminaries}

\textbf{Notation:} We employ $\mc{L}^n_\infty$ to denote the set of bounded signals from time to Euclidean space $x:\R^+\to\R^n$, $\mc{D}^n_\infty$ for bounded discrete time signals $x_d:\N\to\R^n$, and $\Sigma^\omega$ for the set of sequences of infinite length over some alphabet $\Sigma$ of finite cardinality, and $\Sigma^\aleph$ for the space of continuous time signals with image on $\Sigma$, i.e. $x\in\Sigma^\aleph$, $x:\R^+\to\Sigma$. Correspondingly we use $2^X$ for the power set of $X$. We alleviate notation by simply writing $x,y\in\mc{L}_\infty$ to denote that both signals belong to their corresponding $\mc{L}^n_\infty$, with possibly different $n$ for $x$ and $y$.  We denote the identity function by $\mbf{id}$. Given a relation $R\subseteq X\times Y$, we employ the shorthand notation $xRy$ to denote $(x,y)\in R$. The projection map on the $j$-th entry is denoted $\pi_j$, e.g., $\pi_2:(x,y)\mapsto y$. We sometimes abuse notation and denote by $\mcF(A)$ for $\mcF:X\to Y$, $A\in 2^X$ to denote $\mcF(A)=\{y| y=\mcF(x), x\in A\}.$ Finally, for a differential equation: $\dot{\xi}_{u,x}(t)=f(\xi_{u,x}(t),u(t)), \xi(0)=x$, we denote by $\xi_{u,x}(0:\tau)$ the solution to the differential equation, with initial condition $x$ and input $u$, restricted to the time segment $0$ to $\tau$.

\subsection{Problem Statement}

We aim to provide a formal methodology that breaks down a complex synthesis problem into simpler synthesis problems operating at different abstraction layers, such that the composition of these sub-solutions results in satisfaction of the original specification. Moreover, the problem division must enable the independent synthesis of controllers for each sub-problem. To this end, we establish a set of contracts delineating the boundary of operation for each control layer, which decouples the problems addressed at each layer.

We consider systems modelled as complex non-linear ODEs:
\begin{eqnarray*}
    \dot{x}&=&f(x,u)\\
    y&=&h(x),
\end{eqnarray*}
with $x\in\mathcal{X}\subseteq \R^n$, $u\in\mathcal{U}\subseteq \R^m$, and $y\in\mathcal{Y}\subseteq \R^p$.
The goal is to synthesize a controller such that the resulting controlled system satisfies a complex specification provided as, e.g., a temporal logic formula $\phi$ involving atomic propositions defined over the state-space of the system. We illustrate each of the steps, contracts, and controller designs over the following running example.
%Towards defining the type of specifications consider a set of symbols associated to events $p_i := x(t) \in \mathcal{P}_i \subseteq \mathcal{X}$

\begin{example} \label{ex:diffdrive}
 Consider a vehicle  with dynamics:
\begin{equation} \label{eq:diffdrive}
    \begin{bmatrix} \dot{x} \\ \dot{y} \\ \dot{\theta} \\ \dot{v} \\ \dot{\omega} \end{bmatrix} = \begin{bmatrix}
        v \cos \theta \\ v \sin \theta \\ \omega \\ a \\ \alpha
    \end{bmatrix}, \quad \begin{bmatrix}
        a \\ \alpha
    \end{bmatrix} = \begin{bmatrix}
        \frac{m_c d}{I_1} \omega^2 \\ - \frac{m_c d}{I_2} v \omega
    \end{bmatrix}+ \begin{bmatrix}
        \gamma & \gamma \\ \delta & -\delta
    \end{bmatrix} \begin{bmatrix}
        \tau_R \\ \tau_L
    \end{bmatrix}
\end{equation}
in which: $x, y$ denote the Cartesian position, $\theta$ the heading, $v, \omega$ the linear and angular velocity, $a, \alpha$ the linear and angular acceleration, and $\tau_L, \tau_R$ are the left and right motor torques.
Additionally, $m_c, d, I_1, I_2, \gamma, \delta$ are physical constants of the vehicle. Our ultimate objective is to design a controller for \eqref{eq:diffdrive} such that the resulting trajectory satisfies the linear temporal logic (LTL) formula:
\begin{equation}\label{eq:LTL-example}
\begin{aligned}
    \phi = & \square \Diamond \mathrm{base} \wedge \square (\mathrm{base} \implies \bigcirc\neg\mathrm{base} \mathbf{U} \mathrm{gather}) \\
    & \wedge \square (\mathrm{base} \implies \bigcirc\neg\mathrm{base} \mathbf{U} \mathrm{recharge}) \wedge \square\neg\mathrm{danger},
\end{aligned}
\end{equation}
defined over a set of atomic propositions labeling various regions of interest in the environment (Fig. \ref{fig:example}). Designing an end-to-end controller to choose torques to directly satisfy the LTL specification is intractably complex; this motivates the use of a hierarchical control decomposition.
\end{example}

\subsection{Signals and transducers}

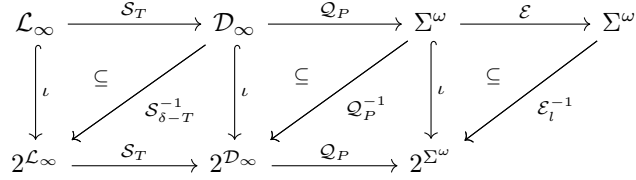
\begin{figure}
    \centering

\begin{tikzcd}[column sep=huge,row sep=huge]
   \mc{L}_\infty\arrow{r}{\mcS_T}\arrow[hookrightarrow]{d}{\iota} 
  & \mc{D}_\infty\arrow{r}{\mc{Q}_P}\arrow[hookrightarrow]{d}{\iota} \arrow{dl}{\mcS_{\delta-T}^{-1}}
  \arrow[swap]{dl}{\subseteq\phantom{aa}}
  & \Sigma^\omega \arrow{dl}{\mc{Q}_P^{-1}}
  \arrow[swap]{dl}{\subseteq\phantom{aa}}
  \arrow[hookrightarrow]{d}{\iota}\arrow{r}{\mc{E}} 
  & \Sigma^\omega \arrow{dl}{\mc{E}_{l}^{-1}}
  \arrow[swap]{dl}{\subseteq\phantom{aa}}
  \\
   2^{\mc{L}_\infty}\arrow{r}[]{\mcS_T} 
  & 2^{\mc{D}_\infty}\arrow{r}[]{\mc{Q}_P}
  & 2^{\Sigma^\omega}
  \end{tikzcd}

    \caption{Sampler, Quantizer, and Eventifier diagram. The diagram commutes up to containment as the inverse transducers are proper.}
    \label{fig:transducers0}
\end{figure}

In layered architectures, we often need to relate systems that represent sets of behaviours of different nature, e.g. discrete-time and continuous-time signals. To facilitate the interfacing of such systems, we define special forms of \textit{transducer} maps (functionals) so that systems interfaced with such maps offer input/output signals compatible with the abstraction level employed at a higher layer. 
%In the remainder of the manuscript, we abuse notation and write $\mc{F}(S)=\lbrace{\mc{F}(x)\,|\, x\in S\rbrace}$ for every map $\mc{F}:X\to Y$ and $S\subseteq X$.

%We employ simply $\mbd(\cdot,\cdot)$ for the usual $\infty$ functional distance (pointwise supremum) for any two elements of $\mc{L}_\infty$ and $\mc{D}_\infty$.
%, and the correspoinding Hausdorff distance when measuring between elements from their powerset.

\begin{definition}[Periodic Sampler] \label{def:F-samp}
The map $\mcS_T:\mc{L}_\infty\to\mc{D}_\infty$, $$\mcS_T(x)=x_d,\,x_d(k)=x(kT),\,\forall\,k\in\N$$ is a \emph{periodic sampler with period $T$}.
\end{definition}
Observe that $\mcS_T$ is surjective but not injective, that is, maps multiple continuous time signals to the same discrete time signal. To define a generalized inverse of the sampler, first we introduce the \textit{linear periodic interpolator}.

\begin{definition}[Linear Periodic Interpolator]\label{def:F_inv-samp}
Given a discrete time signal $x_d:\N\to\R^n$, we denote by $\mathcal{T}_T:\mc{D}_\infty\to\mc{L}_\infty$ the \emph{linear periodic interpolator with period $T$}, such that $\hat{x}=\mathcal{T}_T(x_d)$ where $\forall\, t\in[kT, (k+1)T),\,k\in\N$:
\begin{eqnarray*}
\hat{x}(t)&=&x_d(k)+\frac{t-kT}{T}(x_d(k+1)-x_d(k)).    
\end{eqnarray*}
\end{definition}
The sampler, interpolator, and inverse sampler, which we define next, are general transducers applicable to signals of arbitrary length (finite or infinite). In the case of $x_d$ being a signal of length $1$, i.e., the singelton $x_d(0)$, $\hat{x}(t)=x_d(0),\,\forall\,t\in[0,T)$.

\begin{definition}[$(\epsilon,\delta)-T$ Inverse Sampler]
The $(\epsilon,\delta)-T$ inverse sampler is given by $\mcS_{(\epsilon,\delta)-T}^{-1}:\mc{D}_\infty\to 2^{\mc{L}_\infty}$, $$\mathcal{S}^{-1}_{(\epsilon,\delta)-T}(x_{d}):=\lbrace{x\,|\, \Vert x-\mc{T}_T(x_d) \Vert_\infty\leq \delta \rbrace}\cap\lbrace x\,|\,\Vert x(kT)-x_d(k)\Vert\leq\epsilon\rbrace,$$ the set of bounded continuous functions $(\epsilon,\delta)$ away from the linear periodic interpolator $\mc{T}_{T}(x_d):\R^+\to\R^n$ of $x_d$. 
\end{definition}
If $\epsilon=\delta$ we denote $\mcS_{\delta-T}^{-1}:=\mcS_{(\epsilon,\delta)-T}^{-1}$.
One can devise more general interpolators and derived sampler inverses, e.g., employing splines, or even aperiodic assignment of times. For illustrative purposes, we focus on this simple intuitive interpolator and corresponding inverse sampler in the current manuscript.

\begin{definition}[Quantizer]\label{def:F-quant}
Given a discrete time, $\mathbb{T}=\N$, (or continuous time, $\mathbb{T}=\R$) signal $x:\mathbb{T}\to\mc{X}$, and a partition $P\subset 2^\mc{X}$ of $\mc{X}=\cup_{[\tp_i]\in P}[\tp_i]$, denote by $\tp_i$ the representative of each set $[\tp_i]$, $\tp(x)=\tp$ s.t. $x\in[\tp]$, and $\Sigma=\lbrace\tp_i\rbrace$, the alphabet of symbols of the quantizer.
The map $\mathcal{Q}_P:\mc{D}_\infty\to\Sigma^\omega$ ($\mathcal{Q}_P:\mc{L}_\infty\to\Sigma^\aleph$), %$\mathcal{Q}_P(x)=x_q:\mathbb{T}\to \Sigma$, 
$$\mc{Q}_P(x)(t)=\tp(x(t)),\,\forall\, t\in\mathbb{T},$$    
is the \emph{quantizer with partition $P$}.
\end{definition}
For ease of exposition, we consider in what follows only quantizers for discrete time signals.
Analogously to the sampler, we can define now an inverse quantizer as follows.

\begin{definition}[$P$ Inverse Quantizer]\label{def:F-invquant}
The $P$ inverse quantizer is given by $\mc{Q}_{P}^{-1}:\Sigma^\omega\to 2^{\mc{H}}$, with $\mc{H}$ equal to $\mc{L}_\infty$ or $\mc{D}_\infty$, 
$$\mathcal{Q}^{-1}_{P}(x_{q}):=\lbrace{x_d\,|\, x_d(t)\in[x_q(t)] \rbrace},\; t\in\mathbb{T}.$$ 
\end{definition}

We can go one step further than quantizing space and time, and transform signals of symbols into \emph{discrete events}
% time into event-time. By this we mean transforming the discrete time quantized signals into sequences of \textit{events}, 
defined as a change of symbol of the signal:
\begin{definition}[Eventifier]\label{def:F-event}
% We call an \emph{eventifier} to the map $\mc{E}:\Sigma^\omega\to \Sigma^\omega$, $\mathcal{E}(x_{q}):=x^e_{q}$ with $x^e_{q}$ defined iteratively starting with $k=0,\, r=1,\, x^e_q(0)=x_q(0)$ as:
Given a sequence $x^e_q$ defined iteratively starting with $k=0,\, r=1,\, x^e_q(0)=x_q(0)$ as:
\begin{eqnarray*}
    x^e_q(r)=x_q(k+1),\, r\leftarrow r+1,\, k\leftarrow k+1  & \text{if} & x_q(k+1)\neq x_q(k)\\
    k\leftarrow k+1  & \text{if} & x_q(k+1)=x_q(k).
\end{eqnarray*} 
The map $\mc{E}:\Sigma^\omega\to \Sigma^\omega$, $\mathcal{E}(x_{q}):=x^e_{q}$ is the \emph{eventifier}.
\end{definition}
We can define also a family of inverse eventifiers:
\begin{definition}[$\ell$ Inverse Eventifier]\label{def:F-invevent}
The $\ell$ inverse eventifier is given by $\mc{E}_{\ell}^{-1}:\Sigma^\omega\to 2^{\Sigma_\omega}$, where
$$\mathcal{E}^{-1}_{\ell}(x^e_{q}):=\lbrace{x_q \in \Sigma^\omega_\ell\,|\, \mathcal{E}(x_q)=x^e_q \rbrace},$$ 
and $\Sigma^\omega_\ell$ is the set of signals satisfying 
% $$\min \lbrace{r\in\N\,|\, x_q(k+r)=x_q(k),\,\forall k\in\N\rbrace}\leq \ell.$$
$$\forall k \in \N \quad \exists r(k) \leq l \quad \text{s.t.}\quad  x_q(k+r(k)) \neq x_q(k).$$
\end{definition}
Intuitively, the inverse of an eventified signal $x^e_q$ is the set of all signals that have symbol changes separated by at most $\ell$ time steps, that, when eventified, result in $x^e_q$. We now employ all the defined maps to define a relation between signals of different types:
\begin{definition}[Signal relations]
Given %a map between two spaces $\mathcal{F}:\mc{X}\to\mc{Y}$ and an 
a map $\mathcal{F}:\mc{Y}\to 2^\mc{X}$, the relation $\sim_{\mc{F}} \subseteq \mc{X}\times \mc{Y}$ is given by $x \sim_{\mc{F}} y \Leftrightarrow x \in \mc{F}(y)$. 
%Further, we denote  $\sim_{\mcF^{*}}\subseteq \mc{Y}\times \mc{X}$ to the relation  $y\sim_{\mcF^{*}}x\;\Leftrightarrow\; x\sim_{\mcF}y$.  
\end{definition}
% \begin{definition}[Generalized distance]
% Given a map between two metric spaces $\mathcal{F}:\mc{X}\to\mc{Y}$ and an inverse map $\mathcal{F}^{-1}:\mc{Y}\to 2^\mc{X}$, the generalized metric $\mbd_{\mc{F}^{-1}}:\mc{X}\times \mc{Y}\to \R^+$ is given by $\mbd_{\mc{F}^{-1}}(x,y)=\mbd^H_X(x,\mc{F}^{-1}(y))$, with $\mbd^H_X$ the Haussdorff distance in the space $X$.   
% \end{definition}
%Note that this does not define a proper metric as it already violates the symmetric property on typesetting. \todo{Properties of this pseudo-metric??}
Note that samplers, quantizers, eventifiers, and their inverses, are particular types of transducers between signal spaces familiar to most readers. Nothing prevents us from defining other transducers between such spaces, and inverses mappings to their power sets, allowing to define more exotic ways to relate functions of diverse nature. Observe that for all the transducers $\mc{F}$ we defined ($\mcS$, $\mc{Q}$ and $\mc{E}$), the associated inverses are \emph{proper} in the following sense:
\begin{definition}[Proper inverses]
    For a \emph{transducer} $\mc{F}:\mc{X}\to\mc{Y}$, we say $\mc{F}^{-1}:\mc{Y}\to2^\mc{X}$ is a \emph{proper inverse} if and only if:
    $$
    x\in \mc{F}^{-1}\circ\mc{F}(x),\;\text{and}\;y\in \mc{F}\circ\mc{F}^{-1}(y).
    $$
\end{definition}

% We also introduce a particular class of functions that prove useful in our illustrative running example:
% \begin{definition}[$\mathcal{PC}_{n,\delta}$ signals]
% We denote by $\mathcal{PC}_{n,\delta}(\R^n)$ the space of functions $f:\R\to\R^n$ defined piecewise as the concatenation of $f_i:[t_i,t_{i+1})\to\R^n$ which are $n$ times differentiable, and whose $n$-th derivative $f^{(n)}$ at the discontinuity points experiences jumps bounded by $\delta$ as: $\Vert f^{(n)}(t_i^+)-f^{(n)}(t_i^+)\Vert_2\leq \delta$.
% \end{definition}

\subsection{Systems}\label{symbolic_model} %, Relations, and Compositions
%\subsection{Systems}
To provide a common language for comparing systems of a diverse nature, we recall the notion of a \emph{generalized labeled transition system},
% We now recall a notion of {\it system}, 
akin to a Mealy counterpart of the systems as introduced in \cite{Paulo}.

% \begin{definition}[Generalized Labeled Transition Systems]\label{system}
% A system $S$ is a tuple $S=\left(X,X_0,U,\rTo,Y,H\right)$ consisting of:
% a (possibly infinite) set of states $X$; a (possibly infinite) set of initial states $X_0\subseteq{X}$;
% a (possibly infinite) set of inputs $U$;
% a transition relation $\rTo\subseteq X\times U\times X$;
% a set of outputs $Y$;
% and
% an output map $H:X\rightarrow Y$.
% \end{definition}

\begin{definition}[Generalized Labeled Transition Systems]\label{system}
A system $S$ is a tuple $S=\left(X,X_0,U,\rTo,Y,H\right)$ consisting of:
\begin{itemize}
\item a (possibly infinite) set of states $X$; 
\item a (possibly infinite) set of initial states $X_0\subseteq{X}$;
\item a (possibly infinite) set of inputs $U$;
\item a transition relation $\rTo\subseteq X\times U\times X$;
\item a set of outputs $Y$; and
\item an output map $H:\rTo \rightarrow Y$.
\end{itemize}
\end{definition}

A transition $(x,u,x')\in\rTo$ is also denoted by $x\rTo^{u} x'$. 
If $x\rTo^{u} x'$, state $x'$ is called a $u$-successor of state $x$. We denote by $\textbf{Post}_u(x)$ the set of all \mbox{$u$-successors} of a state $x$, and by $U(x)$ the set of inputs $u\in{U}$ for which $\textbf{Post}_u(x)$ is nonempty. We say the system is \emph{deterministic} if for all states $x$ and inputs $u\in U(x)$, it holds that $|\textbf{Post}_u(x)|=1$; if for some state $x$, $U(x)=\emptyset$ we say the system is \emph{blocking}. 
%We say a system is metric if $Y$ is a metric space.
%We denote by $\mathcal{T}(U,Y)$ the set of all systems associated to a set of inputs $U$ and a set of outputs $Y$.
We also employ $\mc{B}(S)\subseteq 2^{U^\omega\times Y^\omega}$ to denote the set of \textit{behaviours} a system may generate, where behaviours are pairs of (possibly infinite) sequences of input and output symbols respecting the system dynamics, i.e., $(u,y)\in U^\omega\times Y^\omega$ is a behaviour of $S$ iff $\exists x\in X^\omega:\,\forall\, k\in\N,\, (x(k),u(k),x(k+1))\in\rTo\,\wedge\,y(k)=H(x(k),u(k),x(k+1))$. 
Finally, we denominate a system as \textit{open} if the cardinality of its input set is larger than one, i.e., $|U|>1$, and we say it is \textit{closed} otherwise. This classification of systems as open and closed is taken for coherence with the terminology later introduced in the context of contract theory.

\subsection{Temporal logics}

In our running example we employ specifications at the highest layer given by temporal logics, in particular LTL:
\begin{definition}[LTL Syntax \cite{ModelChecking}]
    A LTL formula $\phi$ defined over a set of atomic propositions $\mathcal{AP}$ is recursively defined as:
    \begin{eqnarray*}
        \phi = \top\,|\,p\,|\,\phi_1\wedge\phi_2\,|\,\neg\phi\,|\,\bigcirc \phi\,|\,\phi_1 \mathbf{U} \phi_2,
    \end{eqnarray*}
    where $p\in\mathcal{AP}$ is an atomic proposition, $\phi,\phi_1,\phi_2$ are LTL formulas, $\bigcirc$ is the ``next" temporal operator, and $\mathbf{U}$ is the ``until" temporal operator.
\end{definition}
Given the above syntax, one can define other familiar temporal operators, such as ``eventually" $\diamond\phi\coloneqq\top\mathbf{U}\phi$ and ``always" $\square\phi\coloneqq\neg\Diamond\neg\phi$, among others. The semantics of LTL formulas are interpreted over infinite words $\sigma$ in $2^{\mathcal{AP}}$ -- denoted by $\sigma\in(2^{\mathcal{AP}})^\omega$ -- and are defined as follows.
\begin{definition}[LTL Semantics \cite{ModelChecking}]
    The satisfaction of an LTL formula $\phi$ over $\mathcal{AP}$ by a word $\sigma=\sigma_0\sigma_1\sigma_2\dots\in(2^{\mathcal{AP}})^{\omega}$, denoted by $\sigma\models\phi$, is recursively defined as:
    \begin{itemize}
        \item $\sigma\models \top$;
        \item $\sigma\models p$ if and only if $p\in\sigma_0$;
        \item $\sigma\models \phi_1\wedge\phi_2$ if and only if $\sigma\models \phi_1$ and $\sigma\models \phi_2$;
        \item $\sigma\models \neg\phi$ if and only if $\sigma \nvDash \phi$;
        \item $\sigma\models \bigcirc \phi$ if and only if $\sigma_1\sigma_2\sigma_3\dots\models \phi$;
        \item $\sigma\models \phi_1\mathbf{U}\phi_2$ if and only if $\exists j\geq0$ such that $\sigma_j\sigma_{j+1}\sigma_{j+2}\dots\models \phi_2$ and $\sigma_i\sigma_{i+1}\sigma_{i+2}\dots\models \phi_1$ for all $0\leq i < j$.
    \end{itemize}
\end{definition}
Note that any word satisfying an LTL formula $\sigma\models\phi$ can always be written as a finite prefix of symbols $\sigma_{\mathrm{pre}}\in (2^{\mathcal{AP}})^*$ followed by the infinite repetition of a finite suffix $\sigma_{\mathrm{suff}}\in(2^{\mathcal{AP}})^{*}$ so that each satisfying word takes the form $\sigma=\sigma_{\mathrm{pre}}(\sigma_{\mathrm{suff}})^\omega$. 
%\todo{Decide if we need B\"{u}chi acceptance conditions or if knowing that we just need some infinite word of atomic propositions is enough for our purposes (i.e., all we are really doing is passing down the next atomic proposition to the middle layer).}

\subsection{Assume-Guarantee Contract Theory}

Contract theory is a theory for compositional reasoning aimed at facilitating decoupled system design. The idea behind contract theory is to define interfaces supporting the specification of subsystems' functionalities over restricted contexts of operation, and employ these interfaces to establish correctness of the system integration with respect to a system-wide specification. We summarize the main concepts of the general meta-theory of contracts, and focus on the particular case of assume-guarantee (A/G) contracts following~\cite{Benveniste}.

Contract theory starts with the introduction of \textit{components} $M$ in some universe $\mc{M}$, and of abstractions of components, termed \textit{contracts}, $\mc{C}$. Naturally, a notion of composition between components, denoted by $\Vert$, is required. 
% Related to our concept of open and closed compositions, 
In contract theory, components may be \textit{open} if they exhibit means to interact with the outside world and \textit{closed} otherwise. An \emph{environment} for a component $M$ is another component $E$ composable with $M$, such that $M\Vert E$ is closed. Given two components $M, M'$ we say $M$ is a subcomponent of $M'$, denoted $M\leq M'$, if $M$ satisfies all properties of $M'$.
We now introduce the fundamental object of study of this theory:

\begin{definition}[Contracts~\cite{Benveniste}]
A \emph{contract} $\mc{C}=(\mc{E}_c,\mc{M}_c)$ is a pair where:
\begin{itemize}
    \item $\mc{M}_c\subseteq \mc{M}$ is the set of \emph{implementations} of $\mc{C}$, and
    \item $\mc{E}_c\subseteq \mc{M}$ is the set of \emph{environments} of $\mc{C}$.
\end{itemize}
For any pair $(E,M)\in\mc{E}_c\times \mc{M}_c$, $E$ is an environment for $M$. Hence, $M\Vert E$ must be well defined and closed. A contract possessing no implementation is called \emph{inconsistent}. A contract possessing no environment is called \emph{incompatible}. We write: 
$$
M\models^M \mc{C}\;\text{and}\; E\models^E \mc{C}
$$
to express that $M\in\mc{M}_c$ and $E\in\mc{E}_c$, respectively.
\end{definition}
In any concrete theory of contracts, such contracts must be expressed in some finite way.
In A/G contract theory, contracts are expressed as behavioural specifications. 
\begin{definition}[A/G contracts~\cite{Benveniste}]
An A/G-contract is a pair $\mc{C}=(A,G)$ of assertions, called the assumptions and the guarantees. The set $\mc{E}_c$ of the legal environments for $\mc{C}$ collects all components $E$ such that $E\subseteq A$. The set $\mc{M}_c$ of all components implementing $\mc{C}$ is defined by $A\Vert M\subseteq G$.
\end{definition}
An \textit{environment} for the contract is a component $E$ such that $E\models A$, and an implementation of the contract is any component $M$ such that for all environments $E$ for the contract $M\Vert E \models G$.
Consequently, two contracts $\mc{C}, \mc{C}'$ with $G'\vee \neg A'=G\vee \neg A$ have identical implementations' sets $\mc{M}_c=\mc{M}'_c$ making them equivalent.
Every contract admits an equivalent so called \emph{saturated} contract satisfying $G\vee A =\top$, attained by simply setting $G'=G\vee\neg A$. We say a saturated contract $\mc{C}=(A,G)$ is \textit{consistent} if and only if $G\neq \emptyset$ and \textit{compatible} if and only if $A\neq\emptyset$. A notion of contract \textit{refinement} defining a preorder is 
introduced also in the literature~\cite{Benveniste} which enables the formalization of composition of contracts.
%also introduced:
%\begin{definition}[Contract refinement]
%We say $\mc{C}'$ refines $\mc{C}$, denoted $\mc{C}'\preceq\mc{C}$, whenever any implementation of $\mc{C}'$: (i) implements $\mc{C}$; and (ii) is able to operate in any environment for $\mc{C}$, i.e. $\mc{M}_{c'}\subseteq \mc{M}_c$, $\mc{E}_{c'}\supseteq \mc{E}_c$.    
%\end{definition}

%For A/G contracts, refinement can be stated as:
%$$
%\mc{C}'\preceq\mc{C}\;\text{iff}\; A'\geq A \wedge G'\cup \neg %A'\leq G\cup \neg A.
%$$
%More relevant for our discussion in this manuscript is the notion of composition between two contracts. 
\begin{definition}[Contract composition~\cite{Incer}]
Given contracts $\mc{C}$ and $\mc{C}'$, their composition, written $\mc{C}_1\otimes\mc{C}_2$, is the smallest contract in the refinement order such that:
\begin{itemize}
    \item composing implementations of each contract results in an implementation satisfying the contract composition; and
    \item composing an environment for the composition with an implementation for $\mc{C}_2$ produces an environment for $\mc{C}_1$, and viceversa.
\end{itemize}
\end{definition}
In terms of the components of \emph{saturated} A/G contracts:
$$
\mc{C}_1\otimes\mc{C}_2 = (A_1\cap A_2\cup \neg(G_1 \cap G_2),\;G_1 \cap G_2).
$$
Contract composition is commutative, associative (under compatibility conditions)~\cite{Benveniste}, and monotonic with respect to the refinement preorder~\cite{Incer}.

\section{Vertical system relations}

To relate different systems (vertically), we introduce novel notions of (alternating) (bi)simulation relations, akin to those introduced in \cite{Girard}, but adapted to the Mealy-type systems we manipulate, and generalized to not require the same metric space for both systems being related.

% introduced in \cite{girard,pola1}. 
% First we introduce a notion of (bi)simulation relation, akin to the one introduced in \cite{girard}. Our version is adapted to the Mealy-type systems we manipulate, and generalized to not require the same metric space for both systems being related.

% \begin{definition}[(Approximate) $\mbd$-simulation relation]\label{def:Fsim}

% Let $S_{a}=(X_{a},X_{a0},U_{a},\rTo_{a},Y_a,H_{a})$ and\\ 
% $S_{b}=(X_{b},X_{b0},U_{b},\rTo_{b},Y_b,H_{b})$ be metric systems with 
% %the same output sets $Y_a=Y_b$ and metric $\mathbf{d}$.
% a notion of distance between their outputs $\mbd:Y_a\times Y_b\to\R^+$.
% For $\varepsilon\in\mathbb{R}_0^{+}$, 
% a relation 
% \mbox{$R_\mbd\subseteq X_{a}\times X_{b}$} is said to be an $\varepsilon$-approximate $\mbd$-simulation relation from $S_{a}$ to $S_{b}$ 
% if the following three conditions are satisfied:
% %
% \begin{itemize}
% \item[(i)] for every $x_{a0}\in{X_{a0}}$, there exists $x_{b0}\in{X_{b0}}$ with $(x_{a0},x_{b0})\in{R}_d$;
% %\item[(ii)] for every $(x_{a},x_{b})\in R$, we have \mbox{$\mathbf{d}(H_{a}(x_{a}),H_{b}(x_{b}))\leq\varepsilon$};
% \item[(ii)] for every $(x_a,x_b)\in R_\mbd$, the existence of $x_a\rTo_{a}^{u_a}x'_a$ in $S_a$ implies the existence of $x_b\rTo_{b}^{u_b}x'_b$ in $S_b$ satisfying $(x'_a,x'_b)\in R_\mbd$ and $d(H(x_a\rTo_{a}^{u_a}x'_a), H(x_b\rTo_{b}^{u_b}x'_b))\leq \varepsilon$.
% \end{itemize}  
% \end{definition}

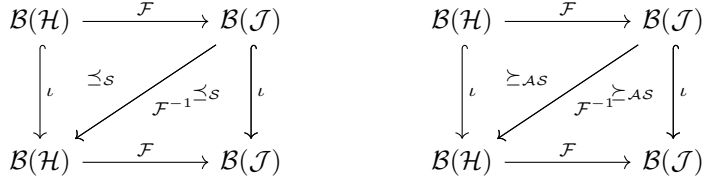
\begin{figure}
\begin{tikzcd}[column sep=huge,row sep=huge]
   \mc{B}(\mc{H}) 
   \arrow{r}{\mc{F}}\arrow[hookrightarrow]{d}{\iota}
   & \mc{B}(\mc{J})
   %\arrow[hookrightarrow]{r}{\iota} 
   \arrow{dl}{\mc{F}^{-1}}
   \arrow[swap]{dl}{{\preceq_\mcS}\phantom{aa}} 
   \arrow[hookrightarrow]{d}{\iota}
   \arrow[hookrightarrow, swap]{d}{{\preceq_\mcS}\phantom{aa}}
   &
    \mc{B}(\mc{H}) 
   \arrow{r}{\mc{F}}\arrow[hookrightarrow]{d}{\iota}
   & \mc{B}(\mc{J})
   %\arrow[hookrightarrow]{r}{\iota} 
   \arrow{dl}{\mc{F}^{-1}}
   \arrow[swap]{dl}{{\succeq_\mcAS}\phantom{a}} 
   \arrow[hookrightarrow]{d}{\iota}
   \arrow[hookrightarrow, swap]{d}{{\succeq_\mcAS}\phantom{a}}
   \\
    \mc{B}(\mc{H})
    %\arrow[hookrightarrow]{r}{\iota} 
   %& \mc{B}(\mc{A})
   \arrow{r}{\mc{F}} %\arrow{u}{\preceq_{?}}
   & \mc{B}(\mc{J})
   &
   \mc{B}(\mc{H})
    %\arrow[hookrightarrow]{r}{\iota} 
   %& \mc{B}(\mc{A})
   \arrow{r}{\mc{F}} %\arrow{u}{\preceq_{?}}
   & \mc{B}(\mc{J})
  \end{tikzcd}    
    \caption{Transducer commutative diagram up to (alternating) simulation. $\mc{H}$ and $\mc{J}$ denote signal spaces related through the transducer $\mc{F}:\mc{H}\to\mc{J}$, and $\mc{B}(\cdot)$ is the space of systems defined over the corresponding signal space.}
    \label{fig:transducers}
\end{figure}

\begin{figure}
\begin{tikzcd}[column sep=huge,row sep=huge]
   S_H \arrow{r}{\mc{F}}\arrow[dashed, tail]{d}{\preceq_{S}/\succeq_{AS}}
   & \mc{F}(S_H)\arrow{r}{\preceq_{S}/\succeq_{AS}} \arrow{dl}{\mc{F}^{-1}} 
  & S_J \arrow{dl}{\mc{F}^{-1}}\arrow[dashed, tail]{d}{\preceq_{S}/\succeq_{AS}}
   \\
    \mc{F}^{-1}\circ\mc{F}(S_H)\arrow{r}{\preceq_{S}/\succeq_{AS}} 
   & \mc{F}^{-1}(S_J)\arrow{r}{\mc{F}} %\arrow{u}{\preceq_{?}}
   & \mc{F}\circ\mc{F}^{-1}(S_J)
  \end{tikzcd}
      
    \caption{Systems, transducers, and relations diagram.}
    \label{fig:sys_rels}
\end{figure}
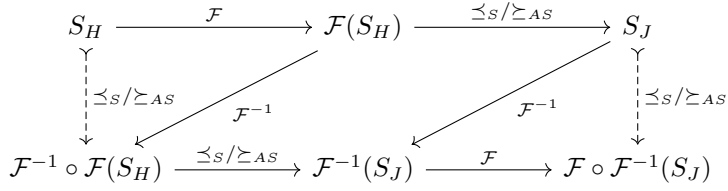

\begin{definition}[$\mcFi$-simulation relation]\label{def:Fsim}
Let $\mcF=\mcF_u\times\mcF_y$, $\mcF_u:U_a\to U_b$, $\mcF_y:Y_a\to Y_b$ be a transducer between two spaces, with corresponding proper inverse $\mcFi=\mcFi_u\times\mcFi_y$, $\mcFi_u:U_b\to2^{U_a}$ and $\mcFi_y:Y_b\to 2^{Y_a}$, and consider two systems
$S_{a}=(X_{a},X_{a0},U_{a},\rTo_{a},Y_a,H_{a})$ and
$S_{b}=(X_{b},X_{b0},U_{b},\rTo_{b},Y_b,H_{b})$. 
%be systems equipped with a notion of equivalence between their outputs $\sim_{\mcFi_y}\subseteq Y_a\times Y_b$.
A relation \mbox{$R_{\mcFi}\subseteq X_{a}\times X_{b}$} is said to be an %$\varepsilon$-approximate 
$\mcFi$-simulation relation from $S_{a}$ to $S_{b}$ 
if the following two conditions are satisfied:
\begin{itemize}
\item[(i)] for every $x_{a}\in{X_{a0}}$, there exists $x_{b}\in{X_{b0}}$ with $(x_{a},x_{b})\in{R}_{\mcFi}$;
\item[(ii)] for every $(x_a,x_b)\in R_{\mcFi}$, for every $u_a\in U(x_a)$, there is $u_b=\mcF_u(u_a)\in U(x_b)$ such that for every $x_a\rTo_{a}^{u_a}x'_a$ in $S_a$ there is $x_b\rTo_{b}^{u_b}x'_b$ in $S_b$ satisfying $(x'_a,x'_b)\in R_{\mcFi}$ and $H_a(x_a\rTo_{a}^{u_a}x'_a) \sim_{\mcFi_y} H_b(x_b\rTo_{b}^{u_b}x'_b)$.
\end{itemize}  
\end{definition}
System $S_{a}$ is %$\varepsilon$-approximately 
$\mcFi$-simulated by $S_{b}$, or $S_b$ %$\varepsilon$-approximately 
$\mcFi$-simulates $S_a$, denoted by \mbox{$S_{a}\preceq_{\mathcal{S},\mcFi}S_{b}$}, if there exists
an %$\varepsilon$-approximate 
$\mcFi$-simulation relation from $S_{a}$ to $S_{b}$. A relation $R_{\mcFi}\subseteq X_a\times X_b$ is said to be an %$\varepsilon$-approximate 
$\mcFi$-bisimulation relation between $S_a$ and $S_b$
if $R_{\mcFi}$ is an %$\varepsilon$-approximate 
$\mcFi$-simulation relation from $S_a$ to $S_b$ and
$R^{-1}_{\mcFi}$ is an %$\varepsilon$-approximate 
$\mcFi$-simulation relation from $S_b$ to $S_a$, where $\sim_{\mcFi}\subseteq Y_b\times Y_a$ such that $y_b\sim_{\mcFi}y_a\;\text{iff}\; y_a\sim_{\mcFi}y_b$.
We denote system $S_{a}$ being %$\varepsilon$-approximately 
$\mcFi$-bisimilar to $S_{b}$ by \mbox{$S_{a}\cong_{\mathcal{S},\mcFi}S_{b}$}, if there exists
an $\mcFi$-bisimulation relation between $S_{a}$ and $S_{b}$.

For nondeterministic systems we need to consider relations that explicitly capture the adversarial nature of nondeterminism. Such relations are necessary to enable the refinement of controllers~\cite{Paulo}. Again, we generalize the definition as earlier to cover different metric spaces and outputs on transitions.

% \begin{definition}[(Approximate) Alternating $\mcF$-simulation relation]\label{def:Faltsim}
% Let $S_{a}=(X_{a},X_{a0},U_{a},\rTo_{a},Y_a,H_{a})$ and\\
% $S_{b}=(X_{b},X_{b0},U_{b},\rTo_{b},Y_b,H_{b})$ be metric systems with with 
% %the same output sets $Y_a=Y_b$ and metric $\mathbf{d}$.
% a notion of distance between their outputs $\mcF:Y_a\times Y_b\to\R^+$.
% For $\varepsilon\in\mathbb{R}_0^{+}$, 
% a relation 
% \mbox{$R\subseteq X_{a}\times X_{b}$} is said to be an alternating $\varepsilon$-approximate $\mcF$-simulation relation from $S_{a}$ to $S_{b}$ 
% if condition (i) %and (ii) 
% in Definition \ref{def:Fsim}, as well as the following condition, are satisfied:
% %
% \begin{itemize}
% %\item[(i)] for every $x_{a0}\in{X_{a0}}$, there exists $x_{b0}\in{X_{b0}}$ with $(x_{a0},x_{b0})\in{R}$;
% %\item[(ii)] for every $(x_{a},x_{b})\in R$ we have \mbox{$\mathbf{d}(H_{a}(x_{a}),H_{b}(x_{b}))\leq\varepsilon$};
% \item[(ii)] for every $(x_a,x_b)\in R$ and for every $u_a\in U_a\left(x_a\right)$ there exists some $u_b\in U_b\left(x_b\right)$ such that for every $x'_b\in\textbf{Post}_{u_b}(x_b)$ there exists $x'_a\in\textbf{Post}_{u_a}(x_a)$ satisfying $(x'_a,x'_b)\in R$
% and $d(H(x_a\rTo_{a}^{u_a}x'_a), H(x_b\rTo_{b}^{u_b}x'_b))\leq \varepsilon$.
% \end{itemize}  
% \end{definition}

\begin{definition}[Alternating $\mcFi$-simulation relation]\label{def:Faltsim}
Let $\mcF=\mcF_u\times\mcF_y$, $\mcF_u:U_a\to U_b$, $\mcF_y:Y_a\to Y_b$ be a transducer between two spaces, with corresponding proper inverse $\mcFi=\mcFi_u\times\mcFi_y$, $\mcFi_u:U_b\to2^{U_a}$ and $\mcFi_y:Y_b\to 2^{Y_a}$, and consider two systems
$S_{a}=(X_{a},X_{a0},U_{a},\rTo_{a},Y_a,H_{a})$,
$S_{b}=(X_{b},X_{b0},U_{b},\rTo_{b},Y_b,H_{b})$. 
%be systems equipped with a notion of equivalence between their outputs $\sim_{\mcFi_y}\subseteq Y_a\times Y_b$.
A relation 
\mbox{$R_{\mcFi}\subseteq X_{a}\times X_{b}$} is said to be an alternating %$\varepsilon$-approximate 
$\mcFi$-simulation relation from $S_{a}$ to $S_{b}$ 
if condition (i) %and (ii) 
in Definition \ref{def:Fsim}, as well as the following condition, are satisfied:
\begin{itemize}
\item[(ii)] for every $(x_a,x_b)\in R_{\mcFi}$ and for every $u_a\in U_a\left(x_a\right)$ there exists some $u_b\in U_b\left(x_b\right)$ such that $u_b=\mcF_u(u_a)$ and for every $x'_b\in\textbf{Post}_{u_b}(x_b)$ there exists $x'_a\in\textbf{Post}_{u_a}(x_a)$ satisfying $(x'_a,x'_b)\in R_{\mcFi}$ and $H_a(x_a\rTo_{a}^{u_a}x'_a) \sim_{\mcF_y^{-1}} H_b(x_b\rTo_{b}^{u_b}x'_b)$.
\end{itemize}  
\end{definition}

System $S_{a}$ is alternatingly %$\varepsilon$-approximately 
$\mcFi$-simulated by $S_{b}$, or $S_b$ alternatingly %$\varepsilon$-approximately 
$\mcFi$-simulates $S_a$, 
denoted by \mbox{$S_{a}\preceq_{\mathcal{AS},\mcFi}S_{b}$}, if there exists
an alternating %$\varepsilon$-approximate 
$\mcFi$-simulation relation from $S_{a}$ to $S_{b}$.
A relation $R_{\mcFi}\subseteq X_a\times X_b$ is said to be an alternating %$\varepsilon$-approximate 
$\mcFi$-bisimulation relation between $S_a$ and $S_b$
if $R_{\mcFi}$ is an alternating %$\varepsilon$-approximate 
$\mcFi$-simulation relation from $S_a$ to $S_b$ and
$R^{-1}_{\mcFi}$ is an alternating %$\varepsilon$-approximate 
$\mcFi$-simulation relation from $S_b$ to $S_a$.
We denote system $S_{a}$ being alternatingly $\mcFi$-bisimilar to $S_{b}$ by \mbox{$S_{a}\cong_{\mathcal{AS},\mcFi}S_{b}$}, if there exists
an alternating %$\varepsilon$-approximate 
$\mcFi$-bisimulation relation between $S_{a}$ and $S_{b}$. Whenever $\mcFi=\mcF=\mbf{id}$ we simply drop the symbol $\mcFi$ from the relations and talk simply about (alternating) (bi)simulation relations.
% If $\varepsilon=0$ we talk about an \textit{exact} (alternating) $\mcF$-(bi)simulation relation and simply drop the $\varepsilon$ from the relation operator.
% System $S_{a}$ is alternatingly $\varepsilon$-approximately bisimilar to $S_{b}$, denoted by \mbox{$S_{a}\cong_{\mathcal{AS},\mcF}^{\varepsilon}S_{b}$}, if there exists
% an alternating $\varepsilon$-approximate $\mcF$-bisimulation relation between $S_{a}$ and $S_{b}$.
% If $\varepsilon=0$ we talk about an \textit{exact} (alternating) $\mcF$-(bi)simulation relation and simply drop the $\varepsilon$ from the relation operator.
We introduce now the notion of transduced systems, that is, systems with their open interfaces affected by transducers to make them compatible with systems defined over different input/output sets.

\begin{definition}[$\mc{F}$-transduced system]
\label{def:F-sys}
    Given a system $S=\left(X,X_{0},U,\rTo,Y,H\right)$ and the transducers 
    $\mc{F}=\mc{F}_u\times\mc{F}_y$, $\mc{F}^{-1}=\mc{F}^{-1}_u\times\mc{F}^{-1}_y$ with $\mc{F}_y:Y\to Y_f$, \mbox{$\mc{F}_u:U\to U_f$}, $\mc{F}_y^{-1}:Y_f\to 2^Y$, and $\mc{F}_u^{-1}:U_f\to 2^U$ we call the \emph{$\mc{F}$-transduced system} to the system: 
    $$
    \mc{F}(S,\mc{F}^{-1})=\left(X,X_{0},U_f,\rTo_f,Y_f,\mc{F}_y\circ H\right),$$ where
    $$
    (x,u_f,x')\in\rTo_f \Leftrightarrow \exists u\in \mc{F}^{-1}_u(u_f)\,\text{s.t.}\,(x,u,x')\in\rTo\text{.}
    $$
\end{definition}
\begin{definition}[$\mc{F}^{-1}$-transduced system]
\label{def:F_inv-sys}
    Given a system $S=\left(X,X_{0},U,\rTo,Y,H\right)$ and the transducers $\mc{F}=\mc{F}_u\times\mc{F}_y$, $\mc{F}^{-1}=\mc{F}^{-1}_u\times\mc{F}^{-1}_y$ with $\mc{F}_y:Y_f\to Y$, \mbox{$\mc{F}_u:U_f\to U$}, $\mc{F}_y^{-1}:Y\to 2^{Y_f}$, and $\mc{F}_u^{-1}:U\to 2^{U_f}$ we call the \emph{$\mc{F}^{-1}$-transduced system} to the system: $$\mc{F}^{-1}(S)=\left(X,X_{0},2^{U_f},\rTo_f,2^{Y_f},\mc{F}^{-1}_y\circ H\right),$$ where %$U_f=2^U$ \todo{or just $U$?}, $Y_f=2^Y$, $H_f:=\mc{F}^{-1}_y\circ H$, and
    $$
    (x,u_f,x')\in\rTo_f \Leftrightarrow \forall u'\in \mc{F}_u(u_f),\; (x,u',x')\in\rTo\text{.}
    $$
\end{definition}
We include $\mc{F}^{-1}$ in the system notation $\mc{F}(S,\mc{F}^{-1})$, as, for a transducer $\mc{F}$, one may define multiple different inverses. We employ simply $\mc{F}(S)$ whenever the inverse is clear from the context. With these definitions in place we can state a few basic results. In each result, $S_a$, $S_b$, and $\mc{F}$ are defined as in Def. \ref{def:Fsim} and \ref{def:Faltsim}.
\begin{proposition}
\label{prop:F_sims}
 $$S_a\preceq_{\mcS} S_b \implies \mcF(S_a) \preceq_{\mcS} \mcF(S_b)$$
 $$S_a\preceq_{\mcAS} S_b \implies \mcF(S_a) \preceq_{\mcAS} \mcF(S_b)$$
\end{proposition}
\begin{proof}
    This is a direct result of applying Definition~\ref{def:F-sys}.
\end{proof}
\begin{proposition}
\label{prop:simsFsims}
    $$S_a\preceq_{\mcSFi} S_b \Leftrightarrow S_a\preceq_{\mcS} \mcFi(S_b)$$
    $$S_b\preceq_{\mcASFi} S_a \Leftrightarrow \mcFi(S_b)\preceq_{\mcAS} S_a$$
\end{proposition}
\begin{proof}
    Provided in the Appendix.
\end{proof}

The previous results demonstrate that simulation relations between two systems hold when each system is $\mcF$-transduced (Prop. \ref{prop:F_sims}) and that there is a one-to-one correspondence between (alternating) $\mcFi$-simulation relations and classical simulations relations (Prop. \ref{prop:simsFsims}). By combining these propositions we can also show that (alternating) $\mcFi$-simulation relations are transitive.
\begin{proposition}[Transitivity]
\label{prop:transitivity}
Given three systems $S_a$, $S_b$, and $S_c$ satisfying:
\begin{eqnarray}\label{eq:transitivity-conditions}
    &S_a \preceq_{\mcS,\mcFi_a} S_b \preceq_{\mcAS,\mcFi_a} S_a&\\
    &S_b \preceq_{\mcS,\mcFi_b} S_c \preceq_{\mcAS,\mcFi_b} S_b,&
\end{eqnarray}
with $\mcFi_a=\mcFi_{u_a}\times\mcFi_{y_a}$, where $\mcFi_{u_a}\,:\,U_b\rightarrow 2^{U_a}$ and $\mcFi_{y_a}\,:\,Y_b\rightarrow 2^{Y_a}$, and $\mcFi_b=\mcFi_{u_b}\times\mcFi_{y_b}$, where $\mcFi_{u_b}\,:\,U_c\rightarrow 2^{U_b}$ and $\mcFi_{y_b}\,:\,Y_c\rightarrow 2^{Y_b}$, define the composite transducers:
\begin{equation*}
\begin{aligned}
     \mcFi_{u_{ab}}\coloneqq \mcFi_{u_a} \circ \mcFi_{u_b},\quad & \mcFi_{u_{ab}}\,:\,U_c\rightarrow 2^{U_a}, \\ 
     \mcFi_{y_{ab}}\coloneqq \mcFi_{y_a} \circ \mcFi_{y_b},\quad & \mcFi_{y_{ab}}\,:\,Y_c\rightarrow 2^{Y_a}. \\ 
\end{aligned}
\end{equation*}
Then:
\begin{equation}\label{eq:transitivity}
    S_a \preceq_{\mcS,\mcFi_{ab}} S_c \preceq_{\mcAS,\mcFi_{ab}} S_a,
\end{equation}
where $\mcFi_{ab}\coloneqq \mcFi_{u_{ab}}\times \mcFi_{y_{ab}}$.
\end{proposition}
\begin{proof}
    Provided in the Appendix.
\end{proof}

In what follows, let $\mathcal{H}$ and $\mathcal{J}$ be two signal spaces related through the transducer $\mcF\,:\,\mathcal{H}\rightarrow\mathcal{J}$, and denote by $\mathcal{B}(\cdot)$ the space of systems defined over the corresponding signal space (see Fig. \ref{fig:transducers}).
\begin{proposition}[Transducers commutative diagram]
\label{prop:trans_cd}
    The diagram in Fig.~\ref{fig:transducers} commutes up to (alternating) simulation, i.e. given systems $S_H\in\mc{B}(\mc{H})$ and $S_J\in\mc{B}(\mc{J})$, a transducer $\mc{F}:\mc{H}\to\mc{J}$, and a corresponding proper inverse $\mc{F}^{-1}:\mc{J}\to 2^\mc{H}$, we have:
    $$
    S_H\preceq_\mcS \mc{F}^{-1}\circ \mc{F}(S_H)\preceq_\mcAS S_H
    $$
    $$
    S_J\preceq_\mcS \mc{F}\circ \mc{F}^{-1}(S_J)\preceq_\mcAS S_J,
    $$
    with the \emph{trivial relation} $R_\mbf{id}\subseteq X\times X$, $xR_\mbf{id}x'\Leftrightarrow x=x'$.
    % Moreover, the diagrams commuting is equivalent to:
    % $$
    % S_H\preceq_{\mcSFi} \mc{F}(S_H) \preceq_{\mcASFi} S_H,
    % $$
    % \todo{Remove if the expression below is not used elsewhere. Seems to not be properly defined, as we only define sims for $\mcFi$ not for regular $\mcF$.}
    % \todo{$$
    % S_J \preceq_{\mcS,\mc{F}} \mc{F}^{-1}(S_J) \preceq_{\mcAS,\mc{F}^*} S_J,
    % $$}
    % also with $R_\mbf{id}$ as relation.
\end{proposition}
\begin{proof}
    Provided in the Appendix.
\end{proof}

% Multiple forms of composition between systems can be devised, one being the familiar feedback composition:
% \begin{definition}[Feedback composition]
% Given two systems $S_1=\left(X_1,X_{01},U_1,\rTo_1,Y_1,H_1\right)$ and $S_2=\left(X_2,X_{02},U_2,\rTo_2,Y_2,H_2\right)$, with $U_1=Y_2$ and $U_2=Y_1$, the feedback composition of $S_1$ with $S_2$, denoted $S_1\Vert S_2$ is given by the system $S_c=\left(X_c,X_{c0},\emptyset,\rTo_c,Y_c,H_c\right)$ with:
% \begin{itemize}
%     \item $X_c=X_1\times X_2$
%     \item $X_{c0}=X_{01}\times X_{02}$
%     \item $\rTo_c = \lbrace ((x_1,x_2),-,(x_1',x_2'))\,|\newline %(x_1,H_1(x_2),x_1')\in\rTo_1 \wedge (x_2,H_2(x_1),x_2')\in\rTo_2 \rbrace$
%     (x_1,u_1,x_1')\in\rTo_1 \wedge (x_2,u_2,x_2')\in\rTo_2 \wedge\newline 
%     u_1=H_2((x_2,u_2,x_2'))\wedge u_2=H_1((x_1,u_1,x_1'))\rbrace$
%     \item $Y_c=Y_1\times Y_2$
%     \item $H_c=H_1\times H_2$
% \end{itemize}
% \end{definition}
% Feedback composition is commutative and results in closed systems, in the sense that no input is available, from the combination of two open systems.

In compositional reasoning, we wish to compose more than two systems towards a higher goal. This motivates notions of composition such that composing two open systems produces another open system. We introduce one such general notion in what follows:
% \begin{definition}[Open feedback composition]
% \label{def:compose}
% \phantom{a}\newline Given two systems $S_1=\left(X_1,X_{01},U_1,\rTo_1,Y_1,H_1\right)$ and \newline $S_2=\left(X_2,X_{02},U_2,\rTo_2,Y_2,H_2\right)$, %with $Y_1=U_2$, 
% and a pair $F=(F_i,F_o)$ of \emph{interconnection maps} $F_i: Y_1 \times Y_2 \times U_c \to U_1\times U_2$ and $F_o: X_1\times X_2 \to Y_c$ the open feedback composition of $S_1$ with $S_2$, denoted $S_1\Vert_F S_2$ is given by the system $S_c=\left(X_c,X_{c0},U_c,\rTo_c,Y_c,H_c\right)$ with:
% \begin{itemize}
%     \item $X_c=X_1\times X_2$
%     \item $X_{c0}=X_{01}\times X_{02}$
%     \item $\rTo_c = \lbrace ((x_1,x_2),u,(x_1',x_2'))\in X_c\times U_c\times X_c\,|\,\newline 
%     \exists (u_1,u_2)=F_i(H_1(x_1,u_1,x_1'),H_2(x_2,u_2,x_2'),u),\; %\newline y_1=H_1(x_1,u_1,x_1'),\;  y_2=H_2(x_2,u_2,x_2')\;
%     \text{s.t.}\newline
%     (x_1,u_1,x_1')\in\rTo_1 \wedge
%     (x_2,u_2,x_2')\in\rTo_2 %\wedge 
%     % \wedge
%     % u_1=H_2((x_2,u_2,x_2'))\wedge u_2=H_1((x_1,F_i(u_1,u),x_1'))
%     \rbrace$
%     \item $H_c=F_o$
% \end{itemize}
% \end{definition}
\begin{definition}[Open sequential feedback composition]
\label{def:compose}
\phantom{a}\newline Given two systems $S_1=\left(X_1,X_{01},U_1,\rTo_1,Y_1,H_1\right)$ and \newline $S_2=\left(X_2,X_{02},U_2,\rTo_2,Y_2,H_2\right)$, %with $Y_1=U_2$, 
and a pair $F=(F_i,F_o)$ of \emph{interconnection maps} $F_i: Y_1 \times Y_2 \times U_c \to U_1\times U_2$, with $F_i=F_{i1}\times F_{i2}$, $F_{i1}:Y_1 \times Y_2\to U_1$, $F_{i2}:Y_1\times U_c\to U_2$, and $F_o: Y_1\times Y_2 \to Y_c$ the open sequential feedback composition of $S_1$ with $S_2$, denoted $S_1\Vert_F S_2$ is given by the system $S_c=\left(X_c,X_{c0},U_c,\rTo_c,Y_c,H_c\right)$ with:
\begin{itemize}
    \item $X_c=X_1\times X_2\times Y_1$
    \item $X_{c0}=X_{01}\times X_{02}\times Y_1$
    \item $\rTo_c = \lbrace ((x_1,x_2,y_1),u_c,(x_1',x_2',y_1'))\in X_c\times U_c\times X_c\,|\,\newline 
    %\newline y_1=H_1(x_1,u_1,x_1'),\;  y_2=H_2(x_2,u_2,x_2')\;
    (x_1,u_1,x_1')\in\rTo_1 \wedge (x_2,u_2,x_2')\in\rTo_2 \wedge\newline 
    y_1'=H_1((x_1,u_1,x_1'))\newline
    \text{with}\; u_1=F_{i1}(y_1, H_2(x_2,u_2,x_2')), u_2=F_{i2}(y_1,u_c)
    % \wedge
    % u_1=H_2((x_2,u_2,x_2'))\wedge u_2=H_1((x_1,F_i(u_1,u),x_1'))
    \rbrace$
    \item $H_c=F_o$
\end{itemize}
\end{definition}
Note that open composition is (in general) not commutative, i.e. $S_1\Vert_F S_2 \neq S_2\Vert_F S_1$, and, in general (depending on $F$), one of the two compositions may not even be well defined. This type of composition is akin to a classical feedback plus feedforward control structure, in which $u\in U_c$ takes the role of a feedforward action, and $S_1$ and $S_2$ act as plant and controller, respectively.
Classical (closed) feedback composition, which we denote by simply $S_1\Vert S_2$, is a particular case when $F=(\mbf{fl}, \pi_1\times \pi_2)$, where $\mbf{fl}(y1,y2,u)=(y_2,y_1)$. Similarly, open compositions also allow to represent sequential compositions of systems with $F_{i1}=\mbf{id}$, $F_{i2}=\pi_2$, $F_o=\pi_1$. As shown in the following result, combining such open compositions with (alternating) $\mcFi$-simulation relations allows to transfer controllers between systems defined over different signal spaces.

\begin{proposition}[Controller transference]
\label{prop:control}
Let the system $S_b=(X_b,X_{b0}, \mcF_u(U_a), \rTo_b, \mcF_y(Y_a), H_b)$ be an abstraction of system $S_a=(X_a,X_{a0}, U_a, \rTo_a, Y_a, H_a)$ in the sense that:
\begin{equation}\label{eq:controller-transfer-abstraction}
    S_a \preceq_{\mcSFi} S_b \preceq_{\mcASFi} S_a,
\end{equation}
with $\mc{F}=\mc{F}_u\times\mc{F}_y$,
$\mc{F}^{-1}=\mc{F}^{-1}_u\times\mc{F}^{-1}_y$ proper,
% $\mc{F}_{y}:\mc{Y}_a\to {\mc{Y}_b}$, $\mc{F}^{-1}_{y}:\mc{Y}_b\to 2^{\mc{Y}_a}$, $\mc{F}_{u}:\mc{U}_a\to {\mc{U}_b}$, $\mc{F}^{-1}_{u}:\mc{U}_b\to 2^{\mc{U}_a}$.
$\mc{F}_{y}:Y_a\to Y_b$, $\mc{F}^{-1}_{y}:Y_b\to 2^{Y_a}$, $\mc{F}_{u}:U_a\to {U_b}$, $\mc{F}^{-1}_{u}:U_b\to 2^{U_a}$. Consider a third system $S_c=(X_c,X_{c0}, U_c, \rTo_c, Y_c, H_c)$ composable with $S_b$ via an interconnection mapping $F$ .
%of the form $F=(F_i, F_o)$ with $F_o=(\pi_1,pi_2)$.
The following holds:
\begin{equation*}
    \mc{F}(S_a)\Vert_{F} S_c \preceq_\mcS \mcF\circ\mcFi(S_b \Vert_{F} S_c) \preceq_\mcAS \mc{F}(S_a)\Vert_{F} S_c.
\end{equation*}
% with $F'=(F_i',F_o')$, $F_i'=F_i\circ (\mbf{id}\times\mbf{id}\times \mc{F}_a)$, $F_o'=\mc{F}^{-1}_a\circ F_o$
\end{proposition}
\begin{proof}
Applying Propositions~\ref{prop:F_sims} and~\ref{prop:simsFsims} to \eqref{eq:controller-transfer-abstraction} results in: $$\mcF(S_a) \preceq_{\mcS} \mcF\circ\mcFi(S_b) \preceq_{\mcAS} \mcF(S_a).$$
Denote by $R_{ab}$ the $\mcFi$-simulation relation from $S_a$ to $S_b$ and $R_{ba}$ the $\mcFi$-alternating simulation relation from $S_b$ to $S_a$.

Let $S_{fc}:=\mc{F}(S_a)\Vert_{F} S_c$, $S_f:=\mc{F}(S_a)$, and $S_{bc}:=\mcF\circ\mcFi(S_b \Vert_{F} S_c)$.
Let $x_{fc,i}$ denote the entries of the state: $x_{fc,1}$, the state of the \emph{plant}, $x_{fc,2}$ the state of the \emph{controller}, and $x_{fc,3}$ the stored output. 
From Definition~\ref{def:compose} we have that $x_{fc}\rTo{u_{fc}}_{fc}x_{fc}'$ if and only if:
$x_{fc,1}\rTo{u_{f}}_{f}x_{fc,1}'$, $x_{fc,2}\rTo{u_{c}}_{c}x_{fc,2}'$,  $x_{fc,3}'=H_f(x_{fc,1}\rTo{u_{f}}_{f}x_{fc,1}')$, with $u_a=F_{i1}\circ H_c(x_{fc,2}\rTo{u_{c}}_{c}x_{fc,2}')$ and 
$u_c=F_{i2}(x_{fc,3}, u_{fc})$.
Let $x_{fc}Rx_{bc}$ if and only if $x_{fc,1}R_{ab}x_{bc,1}$, $x_{fc,2}=x_{bc,2}$, and $x_{fc,3}=x_{bc,3}$.
For every $x_{fc}Rx_{bc}$, $x_{fc}\rTo{u_{fc}}_{fc}x_{fc}'$ we can select $u_{bc}=u_{fc}$, leading to $u_c'=F_{i2}(x_{bc,3}, u_{bc})=F_{i2}(x_{fc,3}, u_{fc})=u_c$ and $$u_b=F_{i1}\circ H_c(x_{bc,2}\rTo{u_c'}x_{bc,2}')=F_{i1}\circ H_c(x_{fc,2}\rTo{u_c}x_{fc,2}')=u_a.$$
As $x_{fc,1}R_{ab}x_{bc,1}$ there exists some $x_{bc,1}\rTo{u_b}x_{bc,1}'$ with $x_{fc,1}'R_{ab}x_{bc,1}'$. $R_{ab}$ is additionally also a simulation relation from $\mcF(S_a)$ to $\mcF\circ\mcFi(S_b)$, and thus $H_{fc}(x_{fc}\rTo{u_{fc}}_f x_{fc}')=$ 
$$(\mcF\circ H_a(x_{fc,1}\rTo_ax_{fc,1}'), H_c(x_{fc,2}\rTo_cx_{fc,2}'))$$
$$\in (\mcF\circ\mcFi\circ H_b(x_{bc,1}\rTo_b x_{bc,1}'), H_c(x_{bc,2}\rTo_c x_{bc,2}').$$
Showing that $R$ is a simulation relation.
%Hence,  there exists $x_{bc}\rTo{u_{fc}}_{bc}x_{bc}'$ such that $x_{fc,1}' R_{fb} x_{bc, 1}'$, $x_{fc,2}'=x_{bc,2}'$ and $x_{fc,3}'=H_b(x_{bc,1}\rTo{u_{f}}_{b}x_{bc,1}')$. 
%Thus, if 
%$H_f(x_{fc,1}\rTo{u_{f}}_{f}x_{fc,1}')\in H_{ffb}(x_{bc,1}\rTo{u_{f}}_{b}x_{bc,1}')$, $x_{fc}'Rx_{bc}'$ and $R$ is a simulation relation from $\mc{F}(S_a)\Vert_{F} S_c$ to $\mcF\circ\mcFi\left(S_b \Vert_{F} S_c\right)$. This holds from the fact that $\mcF(S_a) \preceq_{\mcS} \mcF\circ\mcFi(S_b)$.
The alternating simulation relation follows analogously with the relation: $R'=R^{-1}$, i.e., $x_{bc}R'x_{fc}$ if and only if $x_{bc,1}R_{ba}x_{fc,1}$, $x_{fc,2}=x_{bc,2}$, and $x_{fc,3}=x_{bc,3}$; and employing $\mcF\circ\mcFi(S_b) \preceq_{\mcAS} \mcF(S_a)$.
\end{proof}
\begin{remark}
The transducer $\mcF\circ\mcF^{-1}$ in the sandwich relation enlarges the sets of behaviours of $S_b$ to cover for the case of transducers for which $\mcF\circ\mcF^{-1}\neq\mbf{id}$, as it is the case of e.g., $\mcS_T\circ\mcS^{-1}_{\delta-T}$. 
\end{remark}

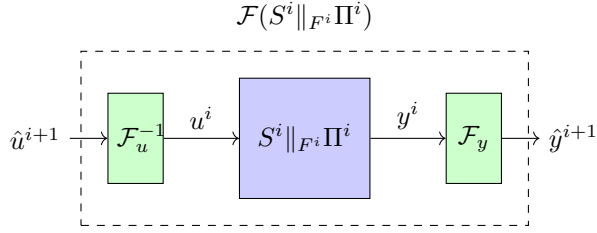
\begin{figure}
\begin{tikzpicture}[auto, node distance=1cm]
\node [rectangle, draw, fill=blue!20, text width=1.5cm, minimum height=1.6cm, text centered] (central) {$S^i\Vert_{F^i}\Pi^i$};

    % Input block
    \node [rectangle, draw, fill=green!20, left=of central, text width=0.5cm, minimum height=1.2cm, text centered] (input) {$\mcFi_u$};
    
    % Output block
    \node [rectangle, draw, fill=green!20, right=of central, text width=0.5cm, minimum height=1.2cm, text centered] (output) {$\mcF_y$};
    
    % Create a dashed box around all blocks
    \node [fit=(input)(central)(output), draw, dashed, inner sep=10pt] (dashedbox) {};
	% Text on top of dashed box
    \node [above=0.15cm of dashedbox, text width=4cm, text centered] (boxLabel) {$\mcF(S^i\Vert_{F^i}\Pi^i)$};

    % Connectors
    \draw [->] (input.east) -- (central.west) node[midway, above] {$u^i$};
    \draw [->] (central.east) -- (output.west) node[midway, above] {$y^i$};
    
     % Arrow from outside to input block
    \node [left=0.5cm of input] (externalInput) {$\hat{u}^{i+1}$};
    \draw [->] (externalInput.east) -- (input.west);
    
    % Arrow from output block to outside
    \node [right=0.5cm of output] (externalOutput) {$\hat{y}^{i+1}$};
    \draw [->] (output.east) -- (externalOutput.west);
\end{tikzpicture}
\caption{Illustration of a layer abstracted for the layer above.\\ 
For $\mcF_u:\mc{U}^i\to\mc{U}^{i+1}$, $\hat{u}^{i+1}\in\mc{U}^{i+1}$, $u^i\subseteq 2^{\mc{U}^i}$,$\mcF_y:\mc{Y}^i\to\mc{Y}^{i+1}$, $\hat{y}^{i}\in\mc{Y}^{i}$, $\hat{y}^{i+1}\in\mc{Y}^{i+1}$.}
\end{figure}

\section{Structured Hierarchical Control}

Hierarchical control, in the form of layered control architectures, is a peculiar form of compositional control systems. Different from \textit{horizontal} compositions in which different systems operate at the same level of abstraction, often acting on different parts of a control plant, e.g., distributed or decentralized controllers, \textit{layered}, also called \textit{hierarchical} or \textit{vertical}, compositions ultimately operate on the same control plant. While the elements of a horizontal composition usually operate at the same abstraction level, in layered architectures subsystems operate (on the same system) at different levels of abstraction. 

To make this idea more concrete, consider a simple plant modeled at the most precise level through ODEs with inputs and outputs $u, y\in\mc{L}_\infty$, and consider an abstract system $S$ capturing the behaviours of such a system $\mc{B}(S)\subset 2^{\mc{L}_\infty \times \mc{L}_\infty}$. 
One may construct a chain of abstractions of these signals by abstracting away at each step fundamental features of the input and output signals by changing the domain and co-domains of these signals. 
The continuous-time signals may be discretized in time, potentially additionally reducing the dimension of the signals, to produce $u_d=\mcS_T(u), y_d=\mcS_T(y) \in \mc{D}_\infty$, thus generating a first level of abstraction. One can conceptualize at this level some system $S_d$ whose behaviours are exactly those time discretized signals $\mc{B}(S_d)\subset 2^{\mc{D}_\infty \times \mc{D}_\infty}$. Next, one may further abstract the signals by, e.g., quantizing (discretizing space) their co-domains. Again, one can conceptualize a system $S_q$ whose behaviours are quantized versions of the behaviours in $\mc{B}(S_d)$, i.e., $u_q=\mc{Q}_P(u), y_q=\mc{Q}_P(y) \in \Sigma^\omega$, $(u_q,y_q)\in \mc{B}(S_q)\subset 2^{\Sigma^\omega \times \Sigma^\omega}$.
%
%Consider the partitioning of both $\mc{U}_d=\cup [\tu_i]$, $[\tu_i]\in 2^{\mc{U}_d}$, and $\mc{Y}_d = \cup [\ty_i]$, $[\ty_i]\in 2^{\mc{Y}_d}$. We denote by $\tu_i$ (respectively $\ty_i$) the representative symbol associated to the subset of inputs $[\tu_i]$ (outputs $[\ty_i]$). Now one can define the next level of abstraction of input and output signals: $u_q:\N\to \mc{U}_q := \lbrace \tu_i \rbrace$, $y_q:\N\to \mc{Y}_q := \lbrace \ty_i \rbrace$. 
%
Note that, while this abstraction chain is the most frequent one encountered in practice, and rather natural, it is still arbitrary. One could think of other chains of abstractions in which continuous-time and space signals are first quantized, and later discretized in time, or abstractions in which not all inputs and outputs are quantized or discretized at the same time. As long as the necessary transducers between signal spaces are properly defined, a chain of abstractions can be established.
In a layered architecture, ultimately, the control problem must be specified in terms of the desired behaviour of the concrete system, i.e., the desired $\cB(S)$. This goal is attained by restricting the set of possible behaviours in an orderly fashion starting from more abstract to more precise characteristics exploiting the chain: 
$$\cB(S)\subseteq \cS_T^{-1}(\cB(S_d)),\; \cB(S_d) \subseteq \cQ_P^{-1}(\cB(S_q)).$$
In other words, the goal of a layered architecture is to remove undesired system behaviours following the hierarchy of abstractions. The difficulty in the design of layered architectures lies thus as much in the concrete definition of the layers as in the mapping of the original specification of desired concrete behaviours %, denoted by $\cB^d(Y)$, 
into specifications of desired behaviours at higher abstraction levels. 
%such that:
%$$\cB^d(Y)\subseteq \cS_\mc{T}^{-1}(\cB^d(Y_d)),\; \cB^d(Y_d)\subseteq \cQ_P^{-1}(\cB^d(Y_q)).$$
%In the next section we establish a structured way of specifying layered solutions as the combination of a set of contracts that employ the transducers and relations between systems established in previous sections.

% We describe now a formal approach to the synthesis of controllers for the satisfaction of a high-level specification by a dynamical system described through a set of (non-linear) ODEs. 
We propose now an approach to the design of layered control in which a high-level specification is refined on a layered set of sub-specifications formalized as contracts. 

Our contracts $\mc{C}_i$ always take the form for each layer $i$:
\begin{eqnarray}
\label{hs:contracts}
\notag
A_i:= &M_{i-1} \wedge &\text{(Model abstraction from layer below)} \\\notag
 &P_{i+1} &\text{(Property of signals of layer above)} \\\notag
%&& P^{i-1}\; \text{(Property of output signals of layer below)}
%Model assumed AND any assumption on the outputs of top layer?
G_i:= &M_{i}  \wedge &\text{(Model abstraction at this layer)} \\
 &P_i &\text{(Property of signals of this layer)}
% guarantee for top layer AND output property for low layer
\end{eqnarray}
On the highest layer $N$, the guarantees do not include a model any more, as there is no higher layer, i.e., $M_{N}=\top$. At the bottom layer -- the most precise model at our disposal of the system under control -- we have $M_0=S_1$, i.e., an abstraction of the real system.
To formalize these properties we make use of the interconnections and system relations previously laid out. In particular, for the model abstraction properties we employ the template:
\begin{equation}
\label{hs:model_abs}
    M_i:= S_i\Vert_{F_i} \Pi_i \preceq_{\mcS,\mc{F}_i^{-1}} S_{i+1} \preceq_{\mcAS,\mc{F}_i^{-1}} S_i\Vert_{F_i} \Pi_i,
\end{equation}
where $\mcF_i$ is a transducer mapping signals from the signal space of $S_i$ to the signal space of $S_{i+1}$, $\mcFi_i$ is the corresponding proper inverse,
% where $\mc{F}_i^{-1}:Y_{i+1}\to 2^{Y_i}$ denotes some inverse transducer of the transducer $\mc{F}_i:Y_{i}\to Y_{i+1}$, 
and $F_i$ is the interconnection map for the controller at layer $i$.

\begin{theorem}\label{thm:composed}    
Consider a set of $N$ contracts of the form~\eqref{hs:contracts}, with model abstraction properties following the template~\eqref{hs:model_abs}. 
The composition $\mc{C}_c=\bigotimes \mc{C}_i = (A_c,G_c)$ is a system-wide contract where:
\begin{eqnarray*}
A_c:=&  S_1 &\text{(Ground-truth model)} \\
G_c:=& \bigwedge_{i=1}^N M_{i}\wedge & \text{(Intermediate models)} \\
& \bigwedge_{i=1}^N P_i &\text{(All properties)} 
\end{eqnarray*}
Moreover, if $S_1$ is correct, then $\bigwedge_{i=1}^N M_{i}$ establishes a formal chain of abstractions so that controllers composed recursively satisfy for all $i\in\{1,\dots,N\}$:
$$
\left(\mcF_i(\tilde{S}_{i}) \Vert_{F_{i+1}} \Pi_{i+1} \right)\models \mcF_i\circ\mcFi_i(P_{i+1}),
$$
$$
\tilde{S}_{i+1}=\left(\mc{F}_{i}(\tilde{S}_{i})\Vert_{F_{i+1}} \Pi_{i+1} \right),\;\tilde{S}_{1}=S_1\Vert_{F_1}\Pi_1.
$$
\end{theorem}
\begin{proof}
   Direct application of contract composition over the equivalent saturated contracts of $\mc{C}_i$ results in $\mc{C}_c$. From the guarantees of this composed contract, direct application of Proposition~\ref{prop:control} leads to the desired result.
\end{proof}
The statement $S\models \mcF\circ\mcFi(P)$ is to be interpreted as: the behaviours of $S$ belong to the set of behaviours that satisfy the property $P$ expanded under $\mcF\circ\mcFi$.

\begin{remark}
Observe that the general theory of contracts, focused on so called \emph{horizontal} compositions, requires the composition between components to be associative and commutative. In the application to hierarchical control, the commutativity requirement needs to be put in context. First, observe that in the open feedback composition we employ interconnections $F_i$ that break commutativity as they determine asymmetrically which of the subsystems gets direct access to the exogenous input. This signal affects only one of the two components directly (the controller).
Fortunately, as we just saw, the order in which components are connected is strictly defined between layers and within layers with systems taking distinctive roles: controller and plants, allowing us to circumvent the commutativity requirements.
\end{remark}

\subsection{A typical 3-layer architecture}\label{sec:typ_layer}

We illustrate here the overall approach on a three layer system with layers determined via sampling and quantization. In Appendix~\ref{app:CLF} another example can be found applying these ideas to a CLF-based controller architecture. 
% In what follows, to avoid too much clutter on the contract descriptions we simply employ ${\mcS^{-1}}$ to refer to ${\mcS_{\delta-T}^{-1}}$, $\mc{Q}^{-1}$ instead of ${\mc{Q}_{P}^{-1}}$, and $\mc{E}^{-1}$ in place of $\mc{E}_\ell^{-1}$. Furthermore, let us define the transducers $\mc{K}:=\mc{E}\circ\mc{Q}$ and $\mc{K}^{-1}:=\mc{Q}^{-1}\circ\mc{E}^{-1}$.
%
We establish the three layers through the assignment of the transducers $\mcF_1:=\mcS$, $\mcF_2:=\mc{K}=\mc{E}\circ\mc{Q}$, with their corresponding proper inverses ${\mcS_{\delta-T}^{-1}}$, and $\mc{K}^{-1}:=\mc{Q}^{-1}\circ\mc{E}^{-1}$.
This results in layer 1 employing ODEs, layer 2 difference equations, and layer 3 FSMs as models. We employ the contracts as just defined and thus we only need to specify now per layer: (i) the corresponding model $S_i$, (ii) the enforced properties at the layer $P_i$. With that information, at each layer the contract satisfaction is achieved through the design of a controller $\Pi_i$ and its corresponding interconnection maps $F_i$.

\subsubsection{Top layer: Model Representation}

At the top layer of control, the goal is to satisfy the high-level specification, operating on long-term strategic decisions. Thus, one may ignore details of how to make the concrete system produce specific symbols (events on the state-space), and simply focus on the ordering of appearance of such symbols (order of events). 

\begin{example}
    The top layer system is represented by a polygonal decomposition of the position space, shown in Fig. \ref{fig:example}, described by a set of symbols $\Sigma$. It is represented by a discrete transition system with transitions between adjacent cells:
    \begin{itemize}
        \item $X_3=X_{3,0} = \Sigma,U_3=\Sigma \times \Sigma, Y_3 = U_3$
        \item $x_3 \rTo_3^{u_3} x_3': x_3 = u_3(1) \wedge x_3' = u_3(2) \wedge\overline{[x_3]}\cap \overline{[x_3']}\neq \emptyset$
        \item $H_3:(x_3, u_3, x_3') \mapsto u_3$
    \end{itemize}
    % \begin{itemize}
    %     \item $X_3=X_{30}=\Sigma$, $U_3=\Sigma\times\Sigma$, $Y_3=X_3$
    %     \item $(x,u,x')\in\rTo_3\iff x=\pi_1(u),\, x'=\pi_2(u),\,\overline{[x]}\cap \overline{[x']}\neq \emptyset$
    %     \item $H_3=(x,u,x')\mapsto x'$.
    % \end{itemize}
    The system $S_3$ operates on quantized space and in event-time.
\end{example}

\subsubsection{Middle layer: Model Representation}

% The middle layer's contract is somewhat more complex in that its environment is determined by both the bottom and top layers. Assumptions need to be made on both the type of requests from the layer above, and on the quality of tracking guaranteed by the lower layer. Similarly, the contract must specify the type of guarantees provided to the lower and higher layers. For the former, 
The middle layer aims to plan trajectories for the system in continuous space, often using a simplified, discrete-time model to facilitate rapid planning. The planned trajectories are passed to the low layer to track, and need to accomplish transitions commanded by the top layer. This allows the middle layer to ignore the complexities of the full-order, nonlinear, continuous-time dynamics, and the complexities of satisfying the LTL specification, making the planning problem tractable.
%; for the latter, the satisfaction of the top layer requests must be guaranteed.
\begin{example}
    The middle layer is represented by an exact discretization of a single integrator with bounded jumps in input between time steps. Here, the state is denoted by $(x_m,y_m,v_x,v_y)\in\R^4$, where $(x_m,y_m)$ represents the planar position of the integrator and $(v_x,v_y)$ its planar velocity. The inputs to the integrator are denoted by $(\Delta v_x,\Delta v_y)\in\R^2$, which correspond to the change in velocity between time-steps. As it may take many time steps to accomplish a high level transition, inputs and outputs of the resulting transition system $S_2=(X_2,X_{20},U_2,\rTo_2,Y_2,H_2)$ are sequences of arbitrary but finite length $n$:% \leq N$:
    % \begin{itemize}
    %     \item $X_2: \begin{bmatrix}
    %         x_m \\ y_m \\ v_x \\ v_y
    %     \end{bmatrix} \in \mathcal{X}_2,U_2: \begin{bmatrix}
    %         \Delta v_x[k] \\ \Delta v_y[k]
    %     \end{bmatrix}_{k=0}^{n-1} \in \mathcal{U}_2^n \subset [0, n] \to \R^2, \\Y_2 = X_2^{n+1}$
    %     \item $x_2 \rTo^{u_2}_{2}x_2': x_{2}[k+1] = \begin{bmatrix}
    %         I & TI \\ 0 & I
    %     \end{bmatrix}x_{2}[k] + \begin{bmatrix}
    %         0 \\ I
    %     \end{bmatrix}u_{2}[k], \\ x_{2}[0]=x_{2},x_{2}[n]=x_{2}'$
    %     \item $H_2:(x_2, u_2, x_2') \mapsto \{x_2[k]\}_{k=0}^{n}$
    % \end{itemize}
    \begin{itemize}
        \item $X_2=\R^4$, $X_{20}=X_2$
        \item $U_2=\{[u[k]]_{k=0}^{n-1}\,:\,u[k]\in\mathcal{U}_2\}$
        \item $(x,u,x')\in\rTo_2\iff$ there exists a sequence $[x[k]]_{k=0}^n$ with $x[k]\in X_2$ such that $x=x[0]$, $x'=x[n]$ and for all $k\in\{0,\dots,n-1\}$: $$x[k+1]=\begin{bmatrix}
            I & TI\\ 0 & I
        \end{bmatrix}x[k] + \begin{bmatrix}
            0 \\ I
        \end{bmatrix}u[k]. $$
        \item $Y_2=\{[x[k]]_{k=0}^n\,:\,x[k]\in X_2\}$
        \item $H_2 = (x,u,x')\mapsto [x[k]]_{k=0}^n$.
    \end{itemize}
    Additionally, let $\mathcal{U}_2 = \{u_2 \in \mathbb{R}^2 : \Vert u_2\|_\infty \leq \Delta v_{max} \}$, i.e., the velocity jumps of the single integrator are bounded. The system $S_2$ operates in a low-dimensional continuous space and discrete-time, with a simple linear model to facilitate rapid planning.
\end{example}

\subsubsection{Low layer: Model Representation}\label{sec:bottomlayer}
The goal of the low layer is to simply guarantee, by properly tracking references from the mid-layer, that at the mid-layer the system can indeed be treated as a linear discrete-time system. The bottom layer uses the full nonlinear dynamics to describe the system.

\begin{example}
    % In our example, let $f$ be the dynamics and $x_1$ to follow the state convention described in \eqref{eq:diffdrive}:  
    % \begin{itemize}
    %     \item $X_1: x_1 \in \mathcal{X}_1 \subset \R^5,U_1: \begin{bmatrix}
    %         \tau_L(t) \\ \tau_R(t)
    %     \end{bmatrix} \in \mathcal{U}_1 \subset [0,\tau) \to \R^2,$ $Y_1 = \{\xi\,:\,[0,\tau)\rightarrow X_1\,|\,\xi(\cdot)\text{ is continuously differentiable}\}$ 
    %     \item $x_1\rTo_1^{u_1} x_1': x_1'=\xi_{u,x_1}(\tau)$, $\dot{\xi}_{u,x}(t)=f(\xi_{u,x_1}(t),u(t)).$
    %     \item $H_1:(x_1, u_1, x_1') \mapsto \xi_{u,x_1}(0:\tau)$
    % \end{itemize}
    In our example, let $f$ be the dynamics characterizing \eqref{eq:diffdrive} and define system $S_1=(X_1,X_{10},U_1,\rTo_{1},Y_1,H_1)$ as:
    \begin{itemize}
        \item $X_1=\mathcal{X}\subset\R^5$, $X_{10}=X_1$
        \item $U_1 = \{u\,:\,[0,\tau)\rightarrow\mathcal{U}\,|\,u(\cdot)\text{ is piecewise continuous}\}$
        \item $(x,u,x')\in\rTo_{1}\iff$ there exists a continuously differentiable function $\xi_{u,x}\,:\,[0,\tau)\rightarrow X_1$ such that $\xi_{u,x}(0)=x$, $\xi_{u,x}(\tau)=x'$, and $\dot{\xi}_{x,u}(t)=f(\xi_{u,x}(t),u(t))$ for all $t\in[0,\tau)$
        \item $Y_1=\{\xi\,:\,[0,\tau)\rightarrow X_1\,|\,\xi(\cdot)\text{ is continuously differentiable}\}$
        \item $H_1=(x,u,x')\mapsto \xi_{u,x}(0:\tau)$,
    \end{itemize}
    where $\mathcal{X}$ and $\mathcal{U}$ as defined as in \eqref{eq:diffdrive}.
    The system $S_1$ operates in both continuous-time and space, with a high fidelity nonlinear model. It takes inputs and generates outputs which are continuous time signals of arbitrary (nonzero) length.
\end{example}

% \textbf{Model abstraction:} At the bottom layer our model abstraction is the most precise model we have available. In the case of models in the form of ODEs, we can formally embed them in a system $S_1=(X_1,X_{10},U_1,\rTo_1, H_1, Y_1)$ with:
% \begin{itemize}
%     \item $X_1=X_{10}:=\mathcal{X}$, $U_1=[0,\tau]\to\mathcal{U}$, $Y_1=[0,\tau]\to\mathcal{Y}$,
%     \item $\rTo_1:=\{(x,u,x')\in \mathcal{X}\times\mathcal{U}\times\mathcal{X}\,|\, x'=\xi_u(\tau)$,\newline \phantom{$\rTo_1:=$aa}$\dot{\xi}_u(t)=f(\xi_u(t),u(t)), \xi(0)=x \}$
%     \item $H_1:=(x,u,x')\mapsto h\circ \xi_{x,u}$,\newline 
%     \phantom{$H_1:=$aa} $\dot{\xi}_{x,u}(t)=f(\xi(t),u(t)), \xi_{x,u}(0)=x$
% \end{itemize}

% \textbf{Property:} \todo{Should we place some property of the signals here? e.g., $P_1:=\mc{B}_y(S_1\Vert_{F_1} \Pi_1)\in\mc{L}_{\infty,D}$, Where we denote by $\mc{L}_{\infty,D}$ the space of functions with $\infty$-norm bounded by $D$.}
\subsubsection{Signal Properties} 
The task remains to design controllers for each layer such that the controlled system satisfies its model abstraction. In general, it is difficult or impossible to make guarantees regarding the controlled system's behavior completely independently of the functionality of the layers above and below it. Signal properties introduce a way to provide additional interfacing information which abstracts system behavior between layers, allowing the controllers to be designed and certified independently of the surrounding layers. 
\begin{example}
    Output and input properties are defined for each system to interface the system with the abstraction above it. Note the superscript $c$ refers to aspects of the controlled system at each layer.
    \begin{itemize}
        \item $P_4 = \top$
        \item High Layer Output Property: $P_3:= S_3\Vert \Pi_3 \models \phi$\newline (high layer output satisfies LTL specification). 
        \item Mid Level Output Property: $P_2: Q_p(y_2^c) \in \mathcal{E}_N^{-1}\circ\mc{E}\circ\mc{Q}_p(y_2^c)$ (mid-layer trajectory generates events in at most $N$ times steps.)
        % \item High Layer Input Property: $P_3^u: u_2^c \in \rTo_{2,c}$ (high layer discrete transitions are feasible for controlled middle layer)
        \item Low Level Output Property: $P_1: y_1^c \in \mathcal{S}_{\delta-T}^{-1}\circ\mc{S}_T(y_1^c)$\newline ($\delta$-bounded tracking of linear interpolation). %of mid-level output signal).
        % \item Middle Layer Input Property: $P_2^u:\Vert u_2 \Vert_\infty \leq \Delta v_{max}$ (bounded changes in velocity of middle layer signal to track at low layer).
    \end{itemize}
    % In terms of the contract template~\eqref{hs:contracts} this results in $P_1=\top$, $P_2=P_2^u$, $P_3=P_3^o\wedge P_3^u$, $P_4=\top$.
\end{example}

\subsubsection{Controller Design}
Finally, all that remains is to design controllers satisfying their corresponding contracts. The hierarchical decomposition of the control task allows each of these designs to be executed independently. 

First, consider the high level controller, which must command transitions of the FSM to enforce some LTL specification (satisfy property $P_3$). 

\begin{example}
    % \todo{Briefly explain/reference techniques to synthesize FSM controller.}
    The FSM $S_3\Vert\Pi_3$ satisfying the LTL specification is shown by blue arrows in Fig. \ref{fig:example}. In general, given the finite transition system $S_3$, a policy $\Pi_3$ enforcing $S_3\Vert\Pi_3\models \phi$ can be computed by translating $\phi$ into a B\"{u}chi/Rabin automaton and solving a reachability game over the product automaton \cite[Ch. 5]{Calin}. For this particular example, such a procedure results in a memoryless state-dependent policy that commands transitions of $S_3$ as indicated in Fig. \ref{fig:example}. That the closed-loop system $S_3\Vert\Pi_3$ satisfies the specification in \eqref{eq:LTL-example} can be seen by noting that an accepting word of $\phi$ is $\sigma=(\mathrm{base}\,\mathrm{recharge}\,\mathrm{gather}\,\mathrm{recharge})^{\omega}$, which corresponds to the transitions shown in Fig. \ref{fig:example} from the given initial condition.
    % \todo{We can just say, that we directly give $S_3\Vert \Pi_3$ as the "resulting controlled" system, which for our specification can be constructed as a memoryless state-dependent policy that just restricts the movements as in the Figure. Note that in the contracts we don't employ $\Pi_3$, the only thing to show is that the closed-loop at that level satisfies the spec, which in this case is clear from visual inspection.}
    The transitions available to the FSM are limited to those the middle layer can achieve; determining which transitions to allow is considered next. 
\end{example}

Next, consider the construction of the middle layer controller. This controller must accomplish the transitions commanded by the high layer (satisfy its model abstraction $M_2$), while ensuring all planned paths satisfy its signal property $P_2$. This is often done by designing MPC controllers to accomplish transitions between areas of the state space. The set of signals satisfying the commanded high-level transition is translated into a set of constraints for the MPC algorithm, which is solved to find a trajectory of the mid-level system representation that satisfies the high-level action.

\begin{example}
     In this case, the $\Pi_2$ system is trivial and simply passes its input to its output, $H_{\Pi_2}(x_{\Pi_2}, u_{\Pi_2}, x_{\Pi_2}') = u_{\Pi_2}$. The open sequential feedback composition is then given by $S_2^c = S_2 \Vert_{F_2} \Pi_2$:
    \begin{itemize}
        \item $X_2^c: X_2 \times X_{\Pi_2} \times Y_2, U_2^c : u_2^c = \mathcal{Q}_P^{-1} \circ \mathcal{E}_N^{-1}(u_3) = \\ \{[x[k], y[k]]_{k=1}^{n \leq N} : \mathcal{E}\circ \mathcal{Q}_P([x[k], y[k]]) = u_3\} \subset 2^{\mathcal{D}_\infty}, \\ Y_2^c:\begin{bmatrix}
            x_m[k] \\ y_m[k]
        \end{bmatrix} \in (\mathcal{Y}_{2}^c)^{n\leq N}$
        
        \item $F_{2,1}=MPC(y_2, y_{\Pi_2}), F_{2,2} = \mathcal{E} \circ \mathcal{Q}_P(u_2^c), F_{2,o} = \begin{bmatrix}
            I & 0
        \end{bmatrix}y_2$
    \end{itemize}
    To demonstrate that the MPC controller satisfies is associated guarantees, first note that it assumes $P_3$ and $M_1$, and define:
    $$\mathcal{C}_{p_i}=\text{maximal control invariant set contained in } p_i \ominus B_\delta$$
    with $\ominus$ the Minkowsi set difference, and $B_\delta$ a ball of radius $\delta$. For a linear system with polytopic constraints, maximal control invariant sets are easily computable offline using Algorithm 10.2 in \cite{Borrelli}. Transitions available to the FSM are only included if the starting partition's control invariant set is $N$-step backwards reachable set of the goal partition's control invariant set:
    \begin{equation*}
        p_i, (p_i, p_j), p_j \in \rTo_{3,c} \quad \iff \quad \mathcal{C}_{p_i} \subset \mathcal{R}_N(\mathcal{C}_{p_j})
    \end{equation*}
    where $\mathcal{R}_N(\cdot)$ computes the $N$-step backwards reachable set of its argument, computable easily offline using methods in Section 10.1.1 of \cite{Borrelli}. During the backwards reachable set computation, the position space is restricted to $(p_i \cup p_j) \ominus B_\delta$ to accommodate for the tracking error of the low level controller. With this restricted definition of the transitions for the LTL solver, the resulting FSM will only command transitions the MPC can achieve. This requires the additional assumption that the initial condition, at time $t=0$, lies within a control invariant set, which, for our system of interest, is not a strong assumption since the control invariant sets are quite large (this amounts to not considering initial conditions close to a partition boundary with significant velocity toward the boundary). $X_{1,0}$ and $X_{2,0}$ must be modified to be consistent with these assumptions.
    
    Given a desired transition $p_i \to p_j$, the MPC problem, outlined in the Appendix \ref{app:mpc}, steers the state from $\mathcal{C}_{p_i}$ into $\mathcal{C}_{p_j}$ in $N$ steps, in our case $N=30$, while guaranteeing that the state remains at least $\delta$ from the boundary of $p_i \cup p_j$. The MPC problem enforces $\mathcal{C}_{p_j}$ as a terminal constraint, guaranteeing that events occur at most $N$ times steps apart, satisfying $P_2$. Bounded jumps in input are enforced are enforced as a constraint in the MPC formulation. It remains to show that the middle layer satisfies its model abstraction, given by
    \begin{equation*}
        M_2:= S_2\Vert_{F_2} \Pi_2 \preceq_{\mcS,\mc{F}_2^{-1}} S_{3} \preceq_{\mcAS,\mc{F}_2^{-1}} S_2\Vert_{F^2} \Pi_2.
    \end{equation*}
    Consider $\mcF_2$ and its proper inverse defined by
    \begin{equation*}
        \mcF_{2,u} = \mathcal{E} \circ \mathcal{Q}_P, \mcF_{2,y} = \mathcal{E} \circ \mathcal{Q}_P, \mcF_{2,u}^{-1} = \mathcal{Q}_P^{-1} \circ \mathcal{E}_N^{-1}, \mcF_{2,y}^{-1} = \mathcal{Q}_P^{-1} \circ \mathcal{E}_N^{-1}
    \end{equation*}
    First, we demonstrate the simulation relation. Consider the relation $R_{\mcFi_2}$ which relates states if the integrator state lies within the control invariant set associated with the partition:
    $$(x_2^c, x_3) \in R_{\mcFi_2} \iff x_2^c(1) \in \mathcal{C}_{x_3},$$
    where $x_2^c(1)=x_2$ is the component corresponding to the state of $S_2$. The initial condition property is trivially satisfied by choosing initial sets properly, as all initial conditions can be quantized $(x_2^c, Q_P(x_2^c(1))) \in R_{\mcFi_2}$. Next, considering an arbitrary $(x_2^c, x_3) \in R_{\mcFi_2}$, picking any $u_2^c \in U_2^c(x_2^c)$, we need to show that $u_3 = \mcF_{2,u}(u_2^c)$ satisfies 
    \begin{multline*}
        \forall x_2^c \rTo_{2,c} ^ {u_2^c} {x_2^c}' \quad \exists x_3 \rTo_{3} ^ {u_3} x_3' \text{ with } \\({x_2^c}', x_3') \in R_{\mcFi_2} \wedge H_2^c(x_2^c\rTo_{2,c}^{u_2^c}{x_2^c}') \sim_{\mcFi_y} H_3(x_3\rTo_{3}^{u_3}x'_3).
    \end{multline*}
    By the simulation relation, $x_2^c(1) \in \mathcal{C}_{x_3}$. Both systems are deterministic and non-blocking, and thus have unique transitions. The MPC formulation guarantees that the transition accomplished by the controlled system via primary and recursive feasibility, $x_2^c \rTo_{2,c} ^ {u_2^c} {x_2^c}'$, results in a state which is contained in a controlled invariant set, ${x_2^c}' \in \mathcal{C}_{p'}$. By the construction of $U_2^c$, the final state of the MPC lies in the same partition as the input $u_2^c$, and therefore $p' = x_3'$, such that $({x_2^c}', x_3') \in R_{\mcFi_2}$. Moreover, $H_2^c$ returns the sequence of positions and $H_3$ returns the initial and final partitions, which must be identical up to quantization and eventification since the terminal states are in the relation. Therefore $H_2^c(x_2^c\rTo_{2,c}^{u_2^c}{x_2^c}') \sim_{\mcFi_y} H_3(x_3\rTo_{3}^{u_3}x'_3)$, implying the simulation relation holds. The alternating simulation relation follows from a similar analysis, a trivial extension since both systems are deterministic.
\end{example}

    Finally, consider the design of the low layer controller. This controller must render the system at sampling times to look like a linear system, without deviating too far from interpolated trajectories. A typical solution consists of first interpolating signals from the middle layer, and then tracking them, e.g., via feedback linearization. 
    
\begin{example}
    The controller system, $\Pi_1$, is responsible for the interpolation; to do so, it must reconstruct the state signal of the middle layer. This is accomplished by having the components of $\Pi_1$ match exactly $S_2$, with the exception of $H_{\Pi_1}$ which performs the interpolation and is defined as:
    $$H_{\Pi_1}:(x_{\Pi_1}, u_{\Pi_1}, x_{\Pi_1}') \mapsto \mathcal{T}_{\alpha, T}([x_m[k]]_{k=0}^n)$$

    The feedback interconnection performs the feedback linearization:
    \begin{itemize}
        \item $X_1^c: X_1 \times X_{\Pi_1} \times Y_1, U_1^c : u_1^c = \mathcal{S}_{0,\delta - T}^{-1}(u_2)  \subset 2^{\mathcal{L}_\infty}, \\Y_1^c:\begin{bmatrix}
            x_l \\ y_l \\ v_{l,x} \\ v_{l,y}
        \end{bmatrix} \in \mathcal{Y}_{1}^c$
        \item $F_{1,1}=k_{fbl}(y_1, y_{\Pi_1}), F_{1,2} = \mathcal{S}_T(u_{\Pi_1}^c), F_{1,o} = y_1$
    \end{itemize}
    The feedback linearizing controller $k_{fbl}$ and interpolation function $\mathcal{T}_{\alpha, T}$ are defined in Appendix~\ref{app:fbl}. 

    To demonstrate that the feedback linearizing controller satisfies its associated guarantees, first note that it assumes only $P_2$ , with no lower layers. The Lyapunov function derived in Appendix~\ref{app:fbl} certifies perfect tracking of the interpolated signal, i.e. $y_1 = x_1 = x_d = \mathcal{T}_{(\alpha, \tau)}([x_m[k]]_{k=0}^n)$. Appendix~\ref{app:fbl} derives the following tracking error bound:
    \begin{equation} \label{eq:trackbound}
        \Vert x_d - \mathcal{T}_{T}(x_m)\Vert_\infty \leq \frac{3}{16}\alpha T\Delta v_{max}.
    \end{equation}
    Taking $\delta = \frac{3}{32}\alpha T\Delta v_{max}$ guarantees $P_1$, implying the low level satisfies its output signal property. Lastly, it remains to show that the low layer satisfies its model abstraction, given by
    \begin{equation*}
        M_1:= S_1\Vert_{F_1} \Pi_1 \preceq_{\mcS,\mc{F}_1^{-1}} S_{2} \preceq_{\mcAS,\mc{F}_1^{-1}} S_1\Vert_{F_1} \Pi_1.
    \end{equation*}
    Consider $\mcF_1$ and its proper inverse defined with
    \begin{equation*}
        \mcF_{1,u} = \mathcal{S}_T,\text{ } \mcF_{1,y} = \mathcal{S}_T,\text{ } \mcF_{1,u}^{-1} = \mathcal{S}_{0-T}^{-1}, \text{ } \mcF_{1,y}^{-1} = \mathcal{S}_{d-T}^{-1}.
    \end{equation*}
    First, we demonstrate the simulation relation. Consider the simulation relation 
    \begin{multline*}
        (x_1^c, x_2) \in R_{\mcFi_1} \iff \\
        \exists \underline{x}_2 \in X_2 \text{ s.t. } x_2 = \begin{bmatrix}
            I & TI \\ 0 & I 
        \end{bmatrix} \underline{x}_2 + \begin{bmatrix}
            0 \\ I
        \end{bmatrix}\Delta v, \Delta v \in \mathcal{U}_2 \wedge \\
        x_1^c(1) = \mathcal{T}_{\alpha, T}(\underline{x}_2, x_2)(T/2) \wedge x_1^c(2) = x_2,
    \end{multline*}
    which equates states of the controlled system with the discrete time integrator if there exists a previous state of the integrator system which, when interpolated, would generate an interpolation passing through the controlled system state. Given any initial condition $x_{2,0} \in X_{2,0}$, convert the initial condition into system $S_1^c$, i.e. $x_{1,0}^c = (x_{1,0}, x_{2,0}, x_{1,0})$, where $x_{1,0}$ has the same position, velocity, and heading as the single integrator, with $\omega=0$ (exclude from $X_{2,0}$ initial conditions with zero velocity to avoid ambiguity with heading of the double integrator). Then $(x_{1,0}^c, x_{2,0}) \in R_{\mcFi_1}$. Under these conditions, the first condition of the simulation relation is satisfied, as there always exists the possible choice 
    $$\underline{x}_2 = \begin{bmatrix}
        I & -TI \\ 0 & I
    \end{bmatrix}x_2, \Delta v = 0.$$
    For the second condition, consider $(x_1^c, x_2) \in R_{\mcFi_1}$. Given any $u_1^c \in U_1^c(x_1^c)$, and the corresponding $u_2 = \mcF(u_1^c)$, both systems (being deterministic and nonblocking), will undergo transitions: 
    $$x_1^c \rTo_{1,c} ^ {u_1^c} {x_1^c}' \quad x_2 \rTo_{2} ^ {u_2} x_2'.$$
    The $\Pi_1$ system will exactly reconstruct the trajectory of the single integrator (by construction) so that $x_{\Pi_1}' = x_2'$, certifying the last condition to be in the simulation relation. As certified by the tracking Lyapunov function, $S_1$ will exactly track the interpolated signal. Considering then ${x_1^c}'$ and $x_2'$, the point 
    $$\underline{x}_2' = \begin{bmatrix}
         I & -TI \\ 0 & I
    \end{bmatrix}x_2', \quad v = u_2(n),$$
    certifies that $({x_1^c}', x_2') \in R_{\mcFi_1}$ where $u_2(n)$ is the last element of the discretized input signal, as $x_1^c$ is arrived at by interpolating between $\underline{x}_2$ and $x_2$. Lastly, the output relation follows from the bounded tracking error established above, as equation \ref{eq:trackbound} implies that $y_1(t) = x_d(t) \in \mathcal{S}_{\delta-T}^{-1}(y_2)$, such that $H_1^c(x_1^c\rTo_{1,c}^{u_1^c}{x_1^c}') \sim_{\mcFi_{1,y}} H_2(x_2\rTo_{2}^{u_2}x'_2)$. The simulation relation is shown; the alternating simulation follows similarly, a trivial extension since the systems are both deterministic. 
\end{example}

\subsubsection{End-to-end correctness}

%\todo{This should be a direct consequence of contract theory applied to the contracts we design}
%\subsection{Contracts between layers}
% Observe that all the contracts defined follow the general structure we proposed with $P^3:= S^3\Vert\Pi^3\models\phi$ and $P^2:=\mc{T}_T(\mc{B}(S^2\Vert_{F^2} \Pi^2))\in\mc{PC}_{r,d}$.
Direct application of contract compositions results in the composed contract $\mc{C}=(S_1, \bigwedge_{i=1}^2 M_i\wedge \bigwedge_{i=1}^3 P_i)$.
%, where: 
% \begin{equation*}
%     M^i:= S^i\Vert_{F^i} \Pi^i \preceq_{\mcS,\mcFi_i} S^{i+1} \preceq_{\mcAS,\mcFi_i} S^i\Vert_{F^i} \Pi^i,
% \end{equation*}
% $M^i:= S^i\Vert_{F^i} \Pi^i \preceq_{\mcS,\mc{F}_i} S^{i+1} \preceq_{\mcAS,\mc{F}_i^*} S^i\Vert_{F^i} \Pi^i$,
% and $\mc{F}_1=\mcS^{-1}$, $\mc{F}_2=\mc{K}^{-1}$. In more practical terms, the composed system satisfies\todo{Revise and fix}
Furthermore, now we can establish from Theorem~\ref{thm:composed} the guaranteed behaviour after composition of the different controllers:
\begin{eqnarray*}
\mc{K}\left(\mc{S}\left(S_1\Vert_{F^1} \Pi_1 \right) \Vert_{F^2} \Pi_2 \right) \Vert_{F^3} \Pi_3 \models \mc{K}\circ\mc{K}^{-1}(P_3),\\
\mc{S}\left(S_1\Vert_{F^1} \Pi_1 \right) \Vert_{F^2} \Pi_2  \models \mc{S}\circ\mc{S}^{-1}(P_2),\\
S_1\Vert_{F^1} \Pi_1 \models P_1.
\end{eqnarray*}
\begin{example} 
As we just showed, each model satisfies its model abstraction and signal properties, allowing the construction of end to end guarantees. 
The correctness of the hierarchical controller designed can be observed in Fig.~\ref{fig:example}. Note that combining the guarantees at each level of abstraction, one could translate the satisfaction of $\mc{K}\circ\mc{K}^{-1}(\phi)$ down to the satisfaction of some signal temporal logic formula for the continuous-time signals of the concrete system ($S_1$). Due to space to limitations we deffer a detailed discussion to future work.
\end{example}

\begin{figure}[h]
    \centering
    \includegraphics[width=0.45\textwidth]{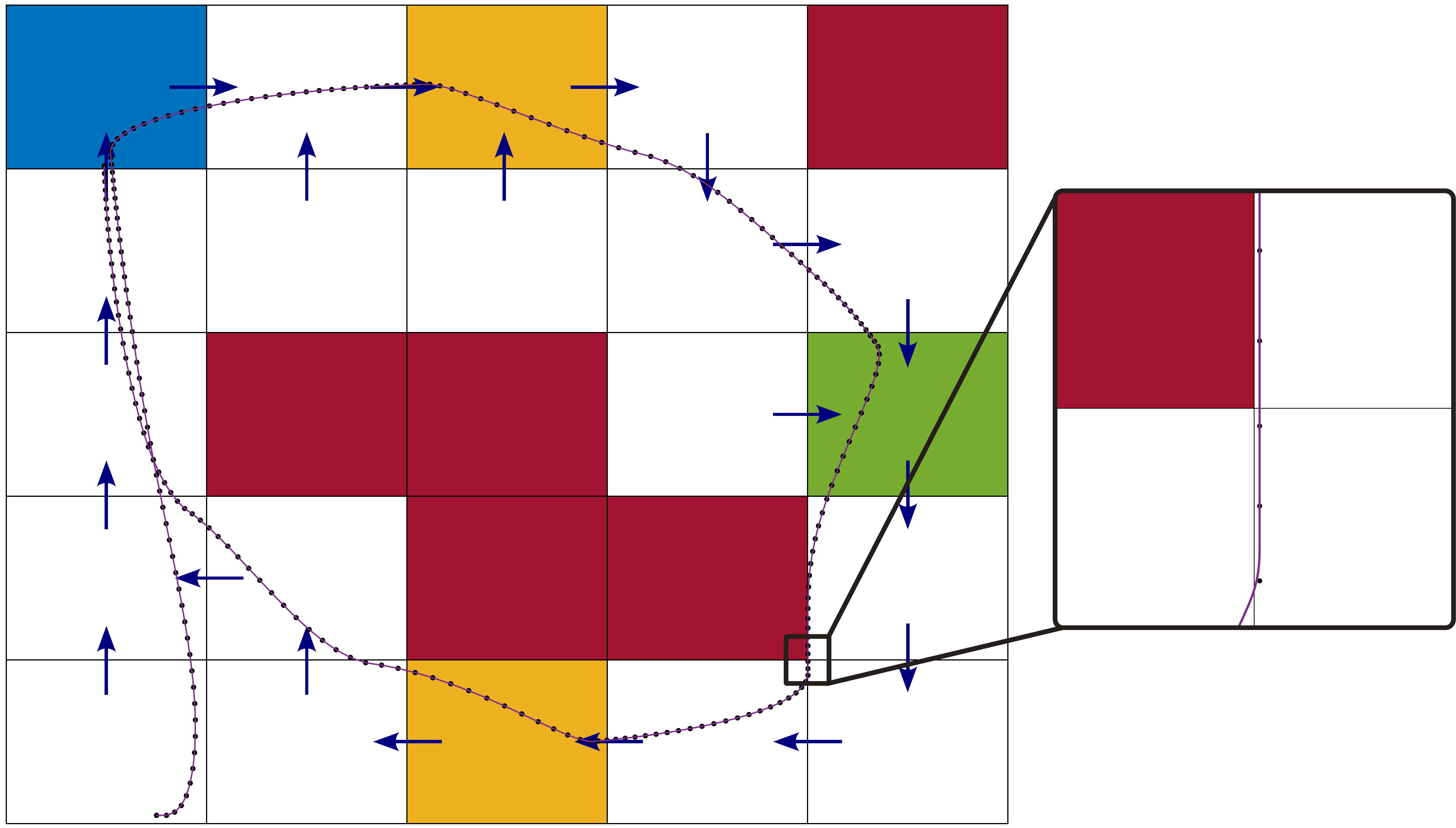}
    \caption{Executing the proposed hierarchical controller. Here, the different colored rectangles represent different propositions from \eqref{eq:LTL-example}, where the blue rectangle corresponds to $\mathrm{base}$, the green rectangle to $\mathrm{gather}$, the yellow rectangles to $\mathrm{recharge}$ and the red rectangles to $\mathrm{danger}$. Arrows encode the FSM's solution to the LTL spec, dots show the MPC plan, and the curve is the path followed by the robot. }
    \label{fig:example}
\end{figure}
\section{Discussion and Conclusions}
We introduced a theory over which to express assume guarantee contracts for hierarchical control architectures. The proposed theory builds on existing work on vertical contracts and simulation relations. In a sense, our proposed system relations for hierarchical control are a an extension of approximate simulation relations~\cite{Girard} to functional approximations. We deliberately used commutative diagrams in our illustrations to hint towards future work on the categorization of layered (control) architectures. The theory, admittedly, can be improved by considering more general (commutative) system interconnections and the combination of both horizontal and vertical contracts. Finally, it would be interesting to explore the potential of the proposed formalization towards architectural (quantitative) optimization, e.g., cross-layer optimizations.
%\newpage

\begin{Appendix}

\section{Proofs}

\begin{proof}[Proof of Proposition 3.6]
We show the equivalence for the simulation relations.
The proof for the alternating version is analogous.
   Let $S_f:=\mcFi(S_b)=(X_b,X_{b0},U_f,\rTo_f, Y_f, H_f)$,
    $H_f:=\mcFi_y\circ H_b$,
    $Y_f= 2^{Y_a}$, and $U_f= 2^{U_a}$.
    Let $R_{ab}\subseteq X_a\times X_b$ be the $\mcFi$-simulation relation from $S_a$ to $S_b$, and $R_{af}\subseteq X_a\times X_f$, the simulation relation $R_{af}=R_{ab}$ from $S_a$ to $S_f$. We focus on condition (ii) of Def.~\ref{def:Fsim}, as condition (i) trivially holds with the given relations both when showing sufficiency and necessity.
    
    Sufficiency: Assume $S_a\preceq_{\mcSFi} S_b$. For every $x_aR_{af}x_f$ and every $x_a\rTo{u_a}_ax_a'$ take $u_f\in\mcFi_u(u_b)$ as the singleton $u_f=\{u\}$, with $u\in U_a$, satisfying $\mcF_u(u)=u_b$.  From Def.~\ref{def:F_inv-sys}, we have $x_b\rTo{u_f}_fx_b'$ 
    % $\Leftrightarrow$ $u'=\mcF_u(u_f)=u_b, x_b\rTo{u_b}_bx_b'$ 
    if and only if for $u'=\mcF_u(u_f)=u_b$ we have $ x_b\rTo{u_b}_bx_b'$, with $x_b=x_f$ and $x_b'=x_f'$.  Moreover, it follows from $S_a\preceq_{\mcSFi} S_b$ that for every $x_b=x_f$ there exists $x_b\rTo{u_b}_bx_b'$, with $x_a'R_{ab}x_b'$, satisfying $H_a(x_a\rTo{u_a}_ax_a')\in \mcFi_y\circ H_b(x_b\rTo{u_b}_bx_b')$. Observing that $\mcFi_y\circ H_b(x_b\rTo{\mcF_u(u_f)}_bx_b')=H_f(x_f\rTo{u_f}_fx_f')$ completes the proof of sufficiency.

    Necessity: Now let $S_a\preceq_{\mcS} \mcFi(S_b)$. For every $x_aR_{ab}x_b$ and every $x_a\rTo{u_a}_a x_a'$ take $u_b=\mcF_u(u_f)$ for the singleton $\{u_f\}$, with $u_f\in U_a$, such that there exists $ x_f\rTo{u_f}_f x_f'$ with $x_a R_{af}x_f$ and $x_a'R_{af}x_f'$, which exists since $S_a\preceq_{\mcS}\mcFi(S_b)$. From Def.~\ref{def:F_inv-sys}, $x_f\rTo{u_f}_fx_f'$ $\Leftrightarrow$ $ x_f\rTo{u_b}_bx_f'$. Moreover, we have $H_a(x_a\rTo{u_a}_a x_a')\in H_f(x_f\rTo{u_f}_fx_f')=\mcFi_y\circ H_b(x_b\rTo{u_b}_bx_b')$, which completes the proof of necessity.
\end{proof}

\begin{proof}[Proof of Proposition 3.7]
    Since $S_a \preceq_{\mcS,\mcFi_a} S_b$ and $S_b \preceq_{\mcS,\mcFi_b} S_c$, it follows from Proposition \ref{prop:simsFsims} that $S_a\preceq_{\mcS}\mcFi_a(S_b)$ and $S_b\preceq_{\mcS}\mcFi_b(S_c)$. Leveraging Proposition \ref{prop:F_sims}, it then follows that $\mcFi_a(S_b)\preceq_{\mcS}\mcFi_a\circ\mcFi_b(S_c)$, which, after using the transitivity of simulation relations \cite{Paulo}, implies that $S_a\preceq_{\mcS}\mcFi_a\circ\mcFi_b(S_c)=\mcFi_{ab}(S_c)$. From Proposition \ref{prop:simsFsims}, this is true if and only if $S_a\preceq_{\mcS,\mcFi_{ab}}S_c$. Using an analogous argument, one can show that if \eqref{eq:transitivity-conditions} holds, then $S_c \preceq_{\mcAS,\mcFi_{ab}}S_a$, which, after combining with the simulation case, implies that \eqref{eq:transitivity} holds, as desired.
\end{proof}

\begin{proof}[Proof of Proposition 3.8]
    First, we show that $\mc{F}^{-1}$ being proper implies that $S_H\preceq_\mcS \mc{F}^{-1}\circ \mc{F}(S_H)$. Let $S_f:=\mc{F}^{-1}\circ \mc{F}(S_H)=(X_H,X_{H0},U_f,\rTo_f, Y_f, H_f)$,
    $H_f:=\mc{F}_y^{-1}\circ\mc{F}_y\circ H_H$,
    $Y_f\subseteq 2^{Y_H}$, and $U_f\subseteq 2^{U_H}$. For
    $u_f\in U_f$, $(x,u_f,x')\in\rTo_f$ if and only if $\forall\,u_m \in \mc{F}_u(u_f)\subseteq 2^{U_H}$, $\exists u'\in \mc{F}^{-1}_u(u_m)$  such that $(x,u',x')\in\rTo_H$, from direct application of Def. \ref{def:F-sys} and \ref{def:F_inv-sys}.
    Note that both $S_H$ and $S_f$ share the same (initial) state sets $X_H$; thus, by employing the trivial relation $R_\mbf{id}$, condition (i) of Def.~\ref{def:Faltsim} is satisfied.
    From $\mc{F}^{-1}$ being proper we know that $u_H\in \mc{F}^{-1}_u\circ\mc{F}_u(u_H)$, and thus for every $(x,u_H,x')\in\rTo_H$ there exists $(x,\{u_H\},x')\in\rTo_f$. Additionally, $H_H((x,u_H,x'))\in \mc{F}_y^{-1}\circ\mc{F}_y\circ H_H((x,\{u_H\},x'))$, and thus condition (ii) also holds.
    
    Next, we show that if $\mc{F}^{-1}$ is proper, $S_H\succeq_\mcAS \mc{F}^{-1}\circ \mc{F}(S_H)$. 
    We consider again the same trivial relation $R_\mbf{id}$ and system $S_f$, implying we only need to show condition (ii) in Def.~\ref{def:Faltsim} holds. From the definition of $S_f$ we have that 
    for every input $u_f$, $x'\in\mbf{Post}^f_{u_f}(x)$ if and only if for all
    $u_m \in \mc{F}_u(u_f)\subseteq 2^{U_H}$ there exists $u'\in \mc{F}^{-1}_u(u_m)$ such that $(x,u',x')\in\rTo_H$.    
    Moreover, for every $u_f\in U_{f}(x)$ and all
    $u_m,u_m'\in u_f$, we have $\mbf{Post}_{\mc{F}(u_m)}(x)=\mbf{Post}_{\mc{F}(u_m')}(x)$ by Def.~\ref{def:F_inv-sys}. We can thus select as $u_m$ any arbitrary representative $u\in u_f$ without affecting the result. 
    Hence, for every $u_f$ we can select in $S_H$ a $u'\in\mc{F}^{-1}_u(u)$ such that $(x,u',x')\in\rTo_H$, and thus $(x,u_f,x')\in\rTo_f$ by definition. As $H_H((x,u_H,x'))\in \mc{F}_y^{-1}\circ\mc{F}_y\circ H_H((x,u_f,x'))$, condition (ii) in Def.~\ref{def:Faltsim} holds.
    The case $S_J\preceq_\mcS \mc{F}\circ \mc{F}^{-1}(S_J) \preceq_\mcAS S_J$ can be shown analogously, employing instead that $u_J\in \mc{F}_u\circ\mc{F}^{-1}_u(u_J)$.
    
    % Finally, we show that the diagram commuting implies $S_H\preceq_{\mcSFi} \mc{F}(S_H)\preceq_{\mcAS,\mc{F}^{-1}} S_H$. As we just showed, the diagrams commuting is equivalent to $S_H\preceq_\mcS \mc{F}^{-1}\circ \mc{F}(S_H)$ and 
    % $S_H\succeq_\mcAS \mc{F}^{-1}\circ \mc{F}(S_H)$ with the trivial relation $R_\mbf{id}$. Let us show first the equivalence for the simulation case.
    % Let $S_\mcF:=\mcF(S_H)=(X_H,X_{H0},U_\mcF,\rTo_\mcF, Y_\mcF, H_\mcF)$,
    % $H_\mc{F}:=\mc{F}_y\circ H_H$,
    % $Y_\mc{F}\subseteq Y_J$, and $U_\mcF\subseteq U_J$. For
    % $u_\mcF\in U_\mcF$, $(x,u_\mcF,x')\in\rTo_\mcF$ if and only if $\exists u_m\in \mc{F}^{-1}_u(u_\mcF)$ such that $(x,u_m,x')\in\rTo_H$.
    % As we just showed $S_H\preceq_\mcS \mc{F}^{-1}\circ \mc{F}(S_H)$, which implies that for all $(x,u,x')\in\rTo_H$, $(x,u,x')\in\rTo_\mcF$ and $H_H((x,u,x'))\in \mc{F}_y^{-1}\circ\mc{F}_y\circ H_H((x,u,x'))\equiv H_H((x,u,x')) \sim_{\mc{F}^{-1}} \mc{F}_y\circ H_H((x,u,x'))\equiv H_H((x,u,x')) \sim_{\mc{F}^{-1}} H_{\mc{F}}((x,u,x'))$. Taking as the relation $R_\mbf{id}$ both conditions (i) and (ii) of Definition~\ref{def:Fsim} hold. The alternating simulation case follows analogously, \todo{and $S_J \preceq_{\mcS,\mc{F}} \mc{F}^{-1}(S_J) \preceq_{\mcAS,\mc{F}} S_J$ can be shown again via duality.} \todo{Verify this last statement}.
\end{proof}

\section{CLF as a 2-layer architecture}
\label{app:CLF}
%\todo{Work in progress}
We now highlight how the familiar use of Control Lyapunov Functions (CLFs) for the design of controllers \cite{Sontag} can also be reinterpreted as a 2-layer control architecture. In this case, following the structure already presented, all we need to do is specify the models employed at the bottom and top layers, and the transducer and generalized inverse relating their output spaces.
Let us start specifying the models. For the bottom layer it takes the same exact form as $S^1$ in the 3-layer example. The top layer model is given by $S^2=(X^2,X^2_0,U^2,\rTo^2, H^2, Y^2)$ with:
\begin{itemize}
    \item $X^2=X^2_0\subseteq \R^+$, $U^2=U^1$, $Y^2=[0,\tau]\to\mathcal{X}^2$,
    \item $\rTo^2:=\{(x,u,x')\in {X}^2\times U^2\times X^2\,|\, x'=\nu_{x,u}(\tau)$,\newline \phantom{$\rTo^1:=$}$\dot{\nu}_{x,u}(t)=g(\nu_{x,u}(t),u(t)), \nu_{x,u}(0)=x \}$ %+d,\; \Vert d\Vert \leq \delta_x  \rbrace}$ 
    \item $H^2:= (x,u,x')\mapsto \nu_{x,u}, \dot{\nu}_{x,u}(t)=g(\nu_{x,u}(t),u(t)),\newline
    \phantom{H^2:= (x,u,x')\mapsto \nu_{x,u}, }\nu_{x,u}(0)=x$
\end{itemize}
To define the transducer between systems consider a CLF $V:X^1\to X^2$, where $X^1\subset \R^{n}$ and $X^2\subset \R^+$. The associated transducer between signals is $\mc{V}:\mc{L}^n_\infty\to\mc{L}_\infty$, $\mc{V}:x(t)\mapsto V(x(t)),\,\forall t$.
%, where $n$ in $\mc{L}^n_\infty$ is used to make explicit the dimension of the image of the functions in the space. 
Define the inverse transducer as $\mc{V}^{-1}:\mc{L}_\infty\to 2^{\mc{L}^n_\infty}$, $\mc{V}^{-1}:v\mapsto \{x\in\mc{L}^n_\infty\,|\, V(x(t))\leq v(t)\,\forall\,t\}$. Now we can define the transducer between systems $\mc{F}=\mc{F}_u\times\mc{F}_y=(\mbf{id},\mc{V})$ and its corresponding inverse $\mc{F}^{-1}=(\mbf{id},\mc{V}^{-1})$.
%driving $\lim_{t\to\infty}v(t)\to 0$. 
Finally, we need to define only the interconnection map at the top layer, as the bottom layer does not incorporate any additional controller. Set $F^2=(\mbf{fl},\pi_2)$, in other words, a standard closed feedback composition with output only the output of the controller.
Let $\Pi^2$ be controller for $S^2$, $S^2\Vert_{F^2}\Pi^2$. The contract that the lower level must satisfy is $A^1:=S^1,\;G^1:=S^1 \preceq_{\mcS,\mcFi} S^2 \preceq_{\mcAS,\mcFi} S^1$.
This contract captures the idea of a Lyapunov function providing an abstraction of the system at hand.
Let the higher level take as contract $A^2=G^1,\;G^2:=\mc{B}(S^2\Vert_{F^2}\Pi^2)\subset \mc{B}_0$. The composition of the contracts guarantees that $\mc{F}(S^1)\Vert_{F^2}\Pi_2 \preceq_{\mcS,\mcFi} S^2\Vert_{F^2}\Pi_2$, which in turn means that the behaviours $\mc{B}(\mc{F}(S^1)\Vert_{F^2}\Pi_2)\subset \mc{F}^{-1}(\mc{B}_0)$. Let $\mc{B}_0$ be the set of behaviours such that $\lim_{t\to\infty}v(t)= 0$. From the definition of $\mc{V}^{-1}$ we have that $\mc{B}(\mc{F}(S^1)\Vert_{F^2}\Pi_2)$ is the set of behaviours such that $\lim_{t\to\infty}{V}(x(t)) = 0$, which, assuming $V$ is positive definite and monotonically increasing implies $\lim_{t\to\infty} x(t) = 0$.
%\todo{May need a bit more explanation.}

% $F^1=(F_{i1}, \mc{V}\circ H^1\circ\pi_1)$ with
% $$F_{i1}: (x^1, x^{\Pi^1}, u) \mapsto (H^{\Pi^1}(x^{\Pi^1}), ((\mbf{sel}(\mc{V}^{-1}),H^{1}(x^1))),$$
% and $F^2=(F_{i2}, H^2\circ\pi_1)$ with
% $$F_{i2}: (x^2, x^{\Pi^2}, u) \mapsto (H^{\Pi^2}(x^{\Pi^2}), (\mbf{id},H^{2}(x^2))).$$
%\todo{to be continued, missing to state how the controller looks like at each level, noting that at the lowest level there's no controller, just a feedthrough from the upper layer}

\section{Differential Drive Robot Example}
\subsection{MPC Formulation} \label{app:mpc}
Consider a feasible commanded transition $p_i \rTo p_j$, and an initial condition $x_{2,0} \in \mathcal{C}_{p_i}$. Each $p_i$ is a polytope, and can be described $A_i x \leq b_i$. For the given spacial partition, the union of $p_i \cup p_j$ is also a convex polytope, and has representation $A_{i,j} x \leq b_{i,j}$. Control invariant sets are also polytopic, $A_{C_{p_i}} x \leq b_{C_{p_i}}$. Let $Q, Q_f \succeq 0, R \succ 0$ be cost matrices, and $x_{2,f}$ be any desired final position (can be set to center of control invariant set, or center of next important partition in FSM solution - base, gather, ect.). The MPC problem is then
\begin{align}
    u_2^* &= \arg\min_{u_2} \Vert x_2[N] - x_{2,f}\Vert_{Q_f}^2 + \sum_{k=0}^{N-1} \Vert x_2[k] - x_{2,f}\Vert_{Q}^2 + \Vert u_2[k]\Vert_{R}^2 \\
    &s.t. \quad x_2[k+1] = \begin{bmatrix}
        I & TI \\ 0 & I
    \end{bmatrix} x_2[k] + \begin{bmatrix}
        0 \\ I
    \end{bmatrix} u_2[k] \\
    &\phantom{s.t. \quad }A_{i,j} x_2[k] \leq b_{i,j} \quad \forall k = 0,\hdots, N-1-n \\
    &\phantom{s.t. \quad }A_{C_{p_j}} x_2[k] \leq b_{C_{p_j}} \quad \forall k = N - n, \hdots, N\\
    &\phantom{s.t. \quad }-\Delta v_{max} \leq u_2[k] \leq \Delta v_{max} \quad \forall k = 0, \hdots, N-1 \\
    &\phantom{s.t. \quad }- v_{max} \leq \begin{bmatrix}
        0 & I
    \end{bmatrix}x_2[k] \leq v_{max} \quad \forall k = 0, \hdots, N \label{eq:mpc_vel}
\end{align}
where $\|s\|_M^2 = s^\top M s$ and \ref{eq:mpc_vel} enforces a velocity constraint, and $n$ is a counter variable. When MPC receives a new transition, $n=0$; this number is incremented every MPC call to enforce the high level transition occurring in at most $N$ times steps. As $C_{p_j}$ is a control invariant set, this problem is recursively feasible, and the high level command is guaranteed to be initially feasible \cite{Borrelli}. 

\subsection{Interpolation Formulation}\label{app:interp}
The interpolation function $\mathcal{T}_{\alpha, T}([x_m[k]]_{k=0}^n)$ takes in a sequence of states of the middle layer system, and interpolates out a dynamically feasible trajectory of positions to be tracked by feedback linearization, using a fourth order Bezier polynomial \cite{bezier}. The $\alpha \in [0, 1]$ governs the proportion of the signal which will be interpolated: $\alpha=0$ corresponds to no interpolation, and $\alpha=1$ interpolates smoothly between midpoints of the signal, and is used in the example. The interpolation is equivalent in the $x$ and $y$ coordinates, so $x$ is done here for simplicity. Consider two points on the trajectory, $x_m[k], x_m[k+1]$, which have $x$ positions and velocities of $s,v$ respectively. The interpolation between these two points is valid over the time interval $t \in [(k+1/2)T, (k+3/2)T)$:
\begin{align}
    x(t) &= \begin{cases}
        x^1(t) & t \in [(k+0.5)T, (k+1)T] \\
        x^2(t) & t \in [(k+1)T, (k + 1.5)T)
    \end{cases} \\
    x^1(t) &= s_k + v_k (t - kT) + \frac{v_{k+1} - v_k}{\delta T^2} t_{p,1}^3\left(1 - 2\frac{t_{p,1}}{\delta T} \right)\\
    x^2(t) &= s_{k+1} + v_{k+1} (t - (k+1)T) + \frac{v_{k+1} - v_k}{\delta T^2} t_{p,2}^3\left(1 - 2\frac{t_{p,2}}{\delta T}\right)\\
    \delta T &= 2(1 - \alpha / 2)T \\
    t_{p,1} &= t - kT - \alpha T /2 \\
    t_{p,2} &= t - kT + \delta T / 2
\end{align}
This interpolation is iterated along the entire trajectory. Endpoints simply use the linear interpolation. Fig. \ref{fig:example} contains the results of applying this interpolation to the MPC trajectory. 

It is necessary to bound the difference between the linear interpolation in position space and this interpolation, to provide an error bound for tracking. Fixing $\alpha=1$ gives $\delta T = T$, and note that both $x^1(t)$ and $x^2(t)$ contain the linear interpolation as their first terms ($s_k + v_k (t - kT)$ and $s_{k+1} + v_{k+1} (t - (k+1)T)$ respectively). Let $\bar{x}$ be the linear interpolation, and consider the error:
\begin{align*}
    x^1(t) - \bar{x} &= \frac{v_{k+1} - v_k}{T^2} t_{p,1}^3\left(1 - 2\frac{t_{p,1}}{T} \right) \\
    x^2(t) - \bar{x} &= \frac{v_{k+1} - v_k}{T^2} t_{p,2}^3\left(1 - 2\frac{t_{p,2}}{T}\right)
\end{align*}
Both functions are maximized at $\frac{T}{2}$, taking the value
\begin{equation}
    |x^1(t) - \bar{x}| = \frac{3}{32} \alpha |v_{k+1} - v_k|T
\end{equation}
Since the change in velocity is bounded, substituting $|v_{k+1} - v_k| \leq \Delta v_{max}$ gives the desired bound. 

\subsection{Feedback Linearization Implementation} \label{app:fbl}
The feedback controller is defined using Lyapunov Backstepping \cite{Krstic}. Let $x_d(t)$ be the desired position signal (three times differentiable), and define error coordinates:
\begin{equation*}
    e = \begin{bmatrix}
    x_l \\ y_l
\end{bmatrix} - x_d(t)
\end{equation*}
The controller is:
\begin{align}
    u &= \frac{1}{2\gamma \delta} \begin{bmatrix}
        \delta & \gamma \\ \delta & -\gamma
    \end{bmatrix} \begin{bmatrix}
        a_d - \frac{m_c d \omega^2}{m + 2I_w/R^2} \\ \alpha_d + \frac{m_cd \omega v}{I + 2L^2 I_w/ R^2}
    \end{bmatrix} \\
    a_d &= \begin{bmatrix}
        \cos\theta & \sin \theta
    \end{bmatrix} \ddot{e}_d\\
    \alpha_d &= -\frac{1}{2} \sigma (\omega - \omega_d) + \dot{\omega_d} - 2 \begin{bmatrix}
        e^\top & \dot{e}^\top
    \end{bmatrix}P \begin{bmatrix}
        0 \\ 0 \\ -v\sin\theta \\ v \cos \theta
    \end{bmatrix} \\
    \ddot{e}_d &= \ddot{x}_d(t) - K_p e - K_d \dot{e} \\
    \omega_d &= \frac{1}{v}\begin{bmatrix}
        -\sin \theta & \cos \theta
    \end{bmatrix} \ddot{e}_d \\
    \dot{\omega}_d &= \begin{bmatrix}
        \frac{a_d}{v^2} \sin \theta - \frac{\omega}{v}\cos \theta & -\frac{a_d}{v^2} \cos \theta - \frac{\omega}{v} \sin \theta 
    \end{bmatrix} \ddot{e}_d  \\ 
    &\phantom{= \text{ }}+\frac{1}{v}\begin{bmatrix}
        -\sin \theta & \cos \theta
    \end{bmatrix} (\dddot{x}_d(t) - K_p \dot{e} - K_d \ddot{e})
\end{align}
Where $P$ solves the CTLE,
\begin{align*}
    A^TP + PA = -I \\
    A = \begin{bmatrix}
        0 & I \\ -Kp & -K_d
    \end{bmatrix}
\end{align*}
The performance of this controller can be certified by construction of a Lyapunov function,
\begin{equation}
    V(e, \dot{e}, \omega, \omega_d) = \begin{bmatrix}
        e^\top & \dot{e}^\top
    \end{bmatrix} P \begin{bmatrix}
        e \\ \dot{e}
    \end{bmatrix} + \frac{1}{2}(\omega - \omega_d)^2
\end{equation}
Under the proposed controller, the time derivative along solutions of the controlled system are:
\begin{align}
    \dot{V}(e, \dot{e}, \omega, \omega_d) &= -\left\|\begin{bmatrix}
        e \\ \dot{e}
    \end{bmatrix} \right\|^2 - \frac{\sigma}{2}(\omega - \omega_d)^2 \\
    &\leq -\lambda V(e, \dot{e}, \omega, \omega_d)
\end{align}
Where $\lambda = \min\{1/\lambda_{max}(P), \sigma\}$. This bound guarantees local exponential tracking of reference signals, and moreover exact tracking of references if the initial condition has zero error \cite{Khalil}. 
\end{Appendix}

\bibliography{biblio}
\bibliographystyle{ieeetr}
\end{document}